\pdfoutput=1

\documentclass{article} 

\usepackage[margin=1in]{geometry}

\usepackage[utf8]{inputenc} 
\usepackage{amsmath} 
\usepackage{amsthm} 
\usepackage{stix2} 
\usepackage{graphicx}
\usepackage{tikz}
\usetikzlibrary{decorations.markings}
\usetikzlibrary{fit}
\usetikzlibrary{3d}
\usetikzlibrary{positioning}
\usetikzlibrary{calc}
\usetikzlibrary{arrows.meta}
\usepackage{array}
\usepackage{enumitem}
\usepackage{hyperref}
\usepackage[nameinlink]{cleveref}
\usepackage{todonotes}
\usepackage{algorithm}
\usepackage{algorithmic}
\newcommand{\BREAK}{\STATE \textbf{break}}

\usepackage{float}
\usepackage{caption}
\usepackage{subcaption}
\usepackage{indentfirst}
\usepackage{authblk}

\tikzset{middlearrow/.style={
        decoration={markings,
            mark= at position 0.5 with {\arrow[scale=1.5]{#1}},
        },
        postaction={decorate},
        >=stealth,
    }
}

\pgfmathsetmacro{\offsetx}{3.7}
\pgfmathsetmacro{\offsety}{2}
\pgfmathsetmacro{\offsetz}{0}

\newtheorem{theorem}{Theorem} 

\newtheorem{lemma}[theorem]{Lemma}
\newtheorem{corollary}[theorem]{Corollary}
\theoremstyle{definition}
\newtheorem{definition}[theorem]{Definition}
\newtheorem{example}[theorem]{Example}

\newtheorem{problem}[theorem]{Problem}
\newtheorem{remark}[theorem]{Remark}
\newtheorem{observation}[theorem]{Observation}

\Crefname{theorem}{theorem}{theorems}
\Crefname{theorem}{Theorem}{Theorems}
\Crefname{lemma}{lemma}{lemmas}
\Crefname{lemma}{Lemma}{Lemmas}
\Crefname{figure}{figure}{figures}
\Crefname{figure}{Figure}{Figures}
\Crefname{problem}{problem}{problems}
\Crefname{problem}{Problem}{Problems}
\Crefname{remark}{remark}{remark}
\Crefname{remark}{Remark}{Remark}
\Crefname{observation}{observation}{observation}
\Crefname{observation}{Observation}{Observation}

\title{Non-Promise Version of Unique Sink Orientations}
\author{Tiago Oliveira Marques \\ \href{mailto:tiago13@mit.edu}{tiago13@mit.edu}}
\affil{Massachusetts Institute of Technology}
\date{}

\begin{document}

\maketitle

\begin{abstract}
A unique sink orientation (USO) is an orientation of the edges of a hypercube such that each face has a unique sink. Many optimization problems like linear programs reduce to USOs, in the sense that each vertex corresponds to a possible solution, and the global sink corresponds to the optimal solution. People have been studying intensively the problem of find the sink of a USO using vertex evaluations, i.e., queries which return the orientation of the edges around a vertex. This problem is a so called \emph{promise} problem, as it assumes that the orientation it receives is a USO.

In this paper, we analyze a \emph{non-promise} version of the USO problem, in which we try to either find a sink or an efficiently verifiable violation of the USO property. This problem is worth investigating, because some problems which reduce to USO are also promise problems (and so we can also define a non-promise version for them), and it would be interesting to discover where USO lies in the hierarchy of subclasses of $\texttt{TFNP}^\texttt{dt}$, and for this a total search problem is required (which is the case for the non-promise version).

We adapt many known properties and algorithms from the promise version to the non-promise one, including known algorithms for small dimensions and lower and upper bounds, like the Fibonacci Seesaw Algorithm. Furthermore, we present an efficient resolution proof of the problem, which shows it is in the search complexity class $\texttt{PLS}^\texttt{dt}$ (although this fact was already known via reductions). Finally, although initially the only allowed violations consist of $2$ vertices, we generalize them to more vertices, and provide a full categorization of violations with $4$ vertices, showing that they are also efficiently verifiable.
\end{abstract}

{\small
\paragraph{Acknowledgements.} I would like to thank the organizers of the Student Summer Research Fellowship at ETH Zürich, for providing me with the opportunity of doing research in Theoretical Computer Science, from which resulted this paper. I would also like to thank my supervisor, Professor Bernd Gärtner, for his guidance through the project. I would like to extend my gratitude to Simon Weber, Sebastian Haslebacher, Patrick Schnider and Michaela Borzechowski for the fruitful brainstorm discussions. Finally, this paper's quality improved significantly after the comments and revision from Simon Weber, Sebastian Haslebacher and Professor Bernd Gärtner, which I deeply appreciate.
}

\section{Introduction}\label{introduction}

In many optimization problems, given a non-optimal solution, it is possible in polynomial time to find a better solution. When modeling this as a directed graph, then an optimal solution is one without outgoing edges, i.e., a sink. Stickney and Watson showed that the graph modelled from the P-matrix Linear Complementarity Problem has a rigid structure, namely the graph is an orientation of the hypercube graph, in which each face has a unique sink \cite{StickneyWatson}. Szabó and Welzl later named these graphs \emph{unique sink orientations} \cite{upperuso}.

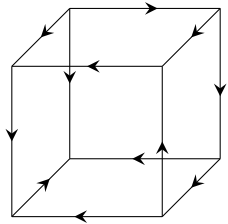
\begin{figure}[htbp]\centering
\begin{tikzpicture}
        \draw[middlearrow={<}] (0,0,0) -- (2,0,0);
        \draw[middlearrow={<}] (2,0,0) -- (2,2,0);
        \draw[middlearrow={<}] (2,2,0) -- (0,2,0);
        \draw[middlearrow={>}] (0,2,0) -- (0,0,0);
        \draw[middlearrow={<}] (0,0,2) -- (2,0,2);
        \draw[middlearrow={>}] (2,0,2) -- (2,2,2);
        \draw[middlearrow={>}] (2,2,2) -- (0,2,2);
        \draw[middlearrow={>}] (0,2,2) -- (0,0,2);
        \draw[middlearrow={<}] (0,0,0) -- (0,0,2);
        \draw[middlearrow={>}] (2,0,0) -- (2,0,2);
        \draw[middlearrow={>}] (0,2,0) -- (0,2,2);
        \draw[middlearrow={>}] (2,2,0) -- (2,2,2);
\end{tikzpicture}
    \caption{A unique sink orientation}
    \label{ex1}
\end{figure}

Many known problems from combinatorics, geometry and mathematical optimization can be reduced to unique sink orientations, where the goal is to find its unique sink \cite{schurrthesis, StickneyWatson, upperuso}. The most studied model to solve such problems is to do \emph{vertex evaluations}, which are queries to an oracle that return the orientation of the edges around a vertex. Sometimes, \emph{edge evaluations}, which are queries to an oracle that return the orientation of a single edge, are also studied, specially for grid USOs, but this is not the focus of this paper \cite[Chapter 3]{grids}. We usually measure the complexity of such algorithms by the number of oracle queries performed. Since in all known reductions to USO vertex evaluations can be implemented in polynomial time, and since all known USO algorithms efficiently calculate the next vertex to query, the true computational complexity is usually only a polynomial factor larger than the query complexity.

USO sink-finding, as introduced above, is a promise problem, as we assume that the oracle corresponds to a unique sink orientation, although checking this property is \texttt{coNP}-complete (when the oracle is represented as a Boolean circuit) \cite[Theorem 5]{gartnerrec}. This means that the algorithm assumes that the input corresponds to a USO, and so we do not care about what it returns when it receives something else.

One can define a non-promise version, where the oracle does not need to encode a USO \cite[Definition 21]{FEARNLEY20201}. The output is allowed to be either a sink or an efficiently verifiable certificate that the oracle does not encode a USO. The certificate should be defined in a way such that every orientation that is not a USO has such certificate. Hence, given any orientation, a solution always exists, either a sink or a certificate, so this problem is a \emph{total search problem}. The existence of such certificates is already known, the simplest one being called a \emph{clash}. A \emph{clash} consists of two distinct vertices which have the same orientation of vertices around it within the face they span in the hypercube graph.

This non-promise version is important because some problems that reduce to USO are also promise problems, so a violation of the USO property can be used as a violation of the initial problem conditions. For example, it is \texttt{coNP}-complete to test whether a given matrix is a P-matrix \cite[Corollary 1]{Coxson1994ThePP}, and so finding the unique solution to a linear complementary problem (LCP) with a P-matrix is a promise problem \cite{StickneyWatson}. Moreover, a non-promise version of the P-matrix LCP has been studied before, where the goal is to find either a solution to the LCP or a nonpositive principal minor of the matrix \cite[Problem 2.3]{nonpromiseplcp}. Hence, if we do not know if a given matrix is a P-matrix, we can still try to solve LCP by trying to find a sink in the hypercube orientation it reduces to, and if a violation is found, then that shows that the matrix is not a P-matrix (although it is not known how a given clash can be used to find a nonpositive principal minor of the matrix).

Some results about the non-promise version of USO sink-finding have been found in complexity theory, namely that it is in the search complexity class $\texttt{UEOPL}^\texttt{dt}$ (the black-box model of Unique End of Potential Line), according to \cite[Subsection 5.2]{endpotentitalline} and \cite[Subsection 5.1]{FEARNLEY20201}, but as far as we know it has not yet been studied from the angle of combinatorial algorithms and query complexity bounds.

We start in \Cref{background} by defining the notation we use and we provide a formal definition of a unique sink orientation. Then, we define the non-promise version of the problem, which will be the main object of study of the paper. Moreover, we also present some known properties and construction techniques which will be useful in the remaining sections.

Afterwards, in \Cref{non-promise}, we show how most results for the promise version still hold for the non-promise one. This includes the $\Omega\left(n^2/\log(n)\right)$ lower bound of Schurr and Szabó \cite[Theorem 9]{quadlower} and the upper bound of $O\left(1.61^n\right)$ by using an adaptation of the Improved Fibonacci Algorithm, originally suggested by Szabó and Welzl \cite[Theorem 4.1]{upperuso}. Then, we show that deterministic algorithms require the same number of queries for both the promise and the non-promise version in dimension up to $4$. However, we show that the same does not hold in the randomized setting, where the non-promise versions require more queries in expectation. While the query complexity is the same for the promise and the non-promise version for $n = 4$, we show that the SevenStepsToHeaven algorithm cannot be adapted to the non-promise setting. We thus introduce a program that finds a new algorithm achieving the same query complexity.

After that, in \Cref{section_resproof}, we investigate the proof complexity of the unsatisfiable CNF formula corresponding to the non-promise version of USO sink-finding. We construct an efficient resolution proof of such CNF formula. The existence of such proof was already known since it is known that USO lies in the search complexity class $\texttt{PLS}^\texttt{dt}$, however no concrete proof was known \cite{endpotentitalline, resproof, FEARNLEY20201}. This resolution proof can possibly provide some insights in proof systems that can efficiently prove the non-promise USO formula, which can be useful if one tries to create a proof system that characterizes the class of problems reducible to USO.

Finally, in \Cref{cert-chapter}, we generalize the notion of a clash to other minimal certificates that show that a given partial orientation (that comes from some vertex evaluations) is not extendable to a USO. Such certificates appear when using an algorithm from the promise version in the non-promise version, because if it fails to find a sink, it must be because the orientation found is not extendable to any USO (but it is possible that no clash was found). Hence, we study such certificates to better understand the difference between the promise and non-promise versions. Schurr already showed that such certificates with $3$ known vertices and no clashes do not exist \cite[Lemma 4.22]{schurrthesis}, and we fully categorize what happens in configurations with $4$ known vertices and no clashes. Moreover, if one understands these certificates, one could try to find upper bounds on the number of queries required to find a clash given any certificate, and that could create upper bounds on the difference between the number of queries required in the promise and non-promise versions. Another possible use of these certificates is to define new non-promise versions of USO sink-finding which accept a larger number of violations beside clashes.

\section{Preliminaries}\label{background}

\subsection{Notation}

We denote $\{1, 2, \dots, n\}$ by $[n]$. We use $\vee, \wedge, \oplus$ to denote the bitwise operations OR, AND and XOR, respectively, when applied to bitstrings. We use $u \cdot v$ as the concatenation of the bitstrings $u$ and $v$. Let $e_i$ be a bitstring with a $1$ in the $i$th position and $0$'s everywhere else, with the appropriate size. Given a bitstring $x$ with $n$ entries and $I \subseteq [n]$, let $x_I$ be the bitstring resulting from $x$ by deleting all the entries with indices not in $I$. Furthermore, $x_i \coloneq x_{\{i\}}$.

Let $\mathfrak{C} \coloneq \{0, 1\}^n$ be a representation of the vertices of an $n$-dimensional hypercube. By abuse of notation, it will also be used to represent the hypercube itself. The edges of $\mathfrak{C}$ consist of $(u, v) \in \mathfrak{C}^2$ such that $u \oplus v = e_i$ for some $i \in [n]$. Call such edges $i$-edges. An orientation $\psi$ of the edges of $\mathfrak{C}$ is represented as $\psi(u, v) = 1$ and $\psi(v, u) = 0$ when the edge is oriented from $u$ to $v$.

\begin{definition}[Face]
    Given an $n$-dimensional hypercube $\mathfrak{C}$, a vertex $u \in \mathfrak{C}$ and a subset of the dimensions $I \subseteq [n]$, let $\mathfrak{C}^I_u$ denote the subgraph of $\mathfrak{C}$ induced by the subset of vertices $\{v \in \mathfrak{C} \mid v_i = u_i, \forall i \in [n] \backslash I\}$. This is called the \emph{face} containing $u$ and spanned by $I$. If $I = [n]$, then $\mathfrak{C}^I_u = \mathfrak{C}$. If $|I| = n-1$, then $\mathfrak{C}^I_u$ is called a \emph{facet}. If $|I| = 1$, then $\mathfrak{C}^I_u$ is called an \emph{edge}. If $|I| = 0$, then $\mathfrak{C}^I_u$ is called a \emph{vertex}.
\end{definition}

\begin{definition}[Antipodal Faces]
    Two faces $\mathfrak{C}_u^I$ and $\mathfrak{C}_v^J$ are called \emph{antipodal} if $I = J$ and $u \oplus 1^n \in \mathfrak{C}_v^J$.
\end{definition}

\begin{definition}[Sink]
    A vertex $v$ in a directed graph $G$ is called a \emph{sink} if it has outdegree $0$, i.e., it has no outgoing edges.
\end{definition}

\subsection{Unique Sink Orientations}

We can now formally define unique sink orientations (USOs).

\begin{definition}[Unique Sink Orientation (USO)]
    An orientation $\psi$ of the edges of $\mathfrak{C}$ is called a \emph{unique sink orientation} if all faces of $\mathfrak{C}$ contain a unique sink.
\end{definition}

For example, \Cref{uso_example} is a USO and \Cref{nonuso-example} is not a USO, as the blue face contains two sinks.

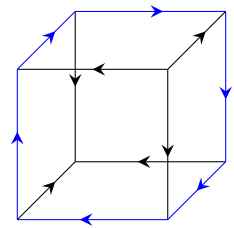
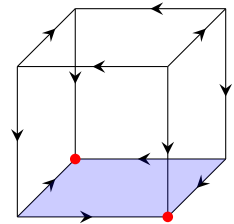
\begin{figure}[htbp]
    \centering
    \begin{subfigure}[t]{0.45\textwidth}
        \centering
        \begin{tikzpicture}
            \draw[middlearrow={<}] (0,0,0) -- (2,0,0);
            \draw[blue, middlearrow={<}] (2,0,0) -- (2,2,0);
            \draw[blue, middlearrow={<}] (2,2,0) -- (0,2,0);
            \draw[middlearrow={>}] (0,2,0) -- (0,0,0);
            \draw[blue, middlearrow={<}] (0,0,2) -- (2,0,2);
            \draw[middlearrow={<}] (2,0,2) -- (2,2,2);
            \draw[middlearrow={>}] (2,2,2) -- (0,2,2);
            \draw[blue, middlearrow={<}] (0,2,2) -- (0,0,2);
            \draw[middlearrow={<}] (0,0,0) -- (0,0,2);
            \draw[blue, middlearrow={>}] (2,0,0) -- (2,0,2);
            \draw[blue, middlearrow={<}] (0,2,0) -- (0,2,2);
            \draw[middlearrow={<}] (2,2,0) -- (2,2,2);
        \end{tikzpicture}
        \caption{A USO with a cycle}
        \label{uso_example}
    \end{subfigure}
    \hfill
    \begin{subfigure}[t]{0.45\textwidth}
        \centering
        \begin{tikzpicture}
            \draw[middlearrow={<}] (0,0,0) -- (2,0,0);
            \draw[middlearrow={<}] (2,0,0) -- (2,2,0);
            \draw[middlearrow={>}] (2,2,0) -- (0,2,0);
            \draw[middlearrow={>}] (0,2,0) -- (0,0,0);
            \draw[middlearrow={>}] (0,0,2) -- (2,0,2);
            \draw[middlearrow={<}] (2,0,2) -- (2,2,2);
            \draw[middlearrow={>}] (2,2,2) -- (0,2,2);
            \draw[middlearrow={>}] (0,2,2) -- (0,0,2);
            \draw[middlearrow={<}] (0,0,0) -- (0,0,2);
            \draw[middlearrow={>}] (2,0,0) -- (2,0,2);
            \draw[middlearrow={<}] (0,2,0) -- (0,2,2);
            \draw[middlearrow={<}] (2,2,0) -- (2,2,2);
            \fill[blue, opacity=0.2] (0,0,0) -- (2,0,0) -- (2,0,2) -- (0,0,2) -- cycle;
            \fill[red] (0,0,0) circle (2pt);
            \fill[red] (2,0,2) circle (2pt);
        \end{tikzpicture}
        \caption{An orientation which is not a USO}
        \label{nonuso-example}
    \end{subfigure}
    \caption{Possible orientations of a $3$-dimensional hypercube.}
\end{figure}

\begin{definition}[Outmap]
    Given an orientation $\psi$ of $\mathfrak{C}$, define the \emph{outmap} of $\psi$ as the function $s \colon \mathfrak{C} \to \{0, 1\}^n$ such that $s(u)_i = \psi(u, u \oplus e_i)$. 
\end{definition}

It is possible that an outmap does not define a proper orientation in $\mathfrak{C}$, as one edge might be oriented in different directions depending on the endpoints. Moreover, if an outmap defines an orientation, then the outmap fully determines that orientation. Furthermore, Szabó and Welzl created a simple condition to check whether an outmap determines a USO (presented later in \Cref{fund}) \cite[Lemma 2.3]{upperuso}. Due to this relation between outmaps and orientations, in this paper sometimes we refer to the outmaps already as a USO.

Unique sink orientations have some properties, which are presented below.

\begin{lemma}\label{flipall}\cite[Lemma 2.1]{upperuso}
    Let $s \colon \mathfrak{C} \to \{0, 1\}^n$ be the outmap of a USO. Define $s' \colon \mathfrak{C} \to \{0, 1\}^n$ as $s'(u) = s(u) \oplus z$ for some $z \in \{0, 1\}^n$. Then, $s'$ is a USO.
\end{lemma}

Therefore, it is possible to give a full categorization of the outmaps of USOs, as follows.

\begin{theorem}\label{fund}\cite[Lemma 2.3]{upperuso}
    An outmap $s \colon \mathfrak{C} \to \{0, 1\}^n$ represents a USO if and only if $(s(u) \oplus s(v)) \wedge (u \oplus v) \neq 0^n$ for all distinct $u, v \in \mathfrak{C}$. Equivalently, $s$ represents a USO if and only if for all distinct $u, v \in \mathfrak{C}$, there exists some $i \in [n]$ such that $u_i \neq v_i \wedge s(u)_i \neq s(v)_i$.
\end{theorem}

As a consequence of this theorem, we call a pair of distinct vertices $u, v \in \mathfrak{C}$ a \emph{clash} if they do not satisfy the property from \Cref{fund}.

\begin{definition}[Clash]\label{clash}
    Given an outmap $s \colon \mathfrak{C} \to \{0, 1\}^n$, a pair of distinct vertices $u, v \in \mathfrak{C}$ is denoted as a \emph{clash} if $(s(u) \oplus s(v)) \wedge (u \oplus v) = 0^n$. 
\end{definition}

As mentioned in \Cref{introduction}, the most classic problem in USOs is to find its global sink by querying an oracle that encodes the outmap of the vertices, in an operation denoted by \emph{query} \cite{upperuso}. The goal is to use the minimum number of queries until the sink is found and evaluated.

\begin{problem}[\textsc{Sink-USO}]\label{originalproblem}\cite{upperuso}
    Given access to an oracle that encodes an outmap $s \colon \mathfrak{C} \to \{0, 1\}^n$ which represents a USO, query the sink, i.e., the vertex $u \in \mathfrak{C}$ such that $s(u) = 0^n$.
\end{problem}

\textsc{Sink-USO} is a promise problem, as it assumes that the input outmap is a USO \cite[Definition 2.4.1]{endpotentitalline}. Moreover, checking this property is \texttt{coNP}-complete (when the oracle is represented by a Boolean circuit) \cite[Theorem 5]{gartnerrec}, and if the oracle is just a black-box, then it is easy to see that $2^n$ queries are required.

Based on \Cref{fund}, Fearnley et al. created a non-promise version of USO sink-finding \cite[Definition 21]{FEARNLEY20201}, which is replicated below.

\begin{problem}[\textsc{Sink-or-Clash}]\label{sinkorclash}
    Given access to an oracle that returns a function $s \colon \mathfrak{C} \to \{0, 1\}^n$, query either
    \begin{enumerate}
        \item a vertex $u \in \mathfrak{C}$ such that $s(u) = 0^n$.
        \item distinct vertices $u, v \in \mathfrak{C}$ such that $(s(u) \oplus s(v)) \wedge (u \oplus v) = 0^n$.
    \end{enumerate}
\end{problem}

\textsc{Sink-or-Clash} receives an outmap (possibly not a USO) and finds either a sink (guaranteed to exist if $s$ is a USO) or a clash (which must exist if $s$ is not a USO). It still uses an oracle to evaluate the outmap $s$, and either the sink or the two vertices in the clash must be evaluated. As a solution always exists, this is a total search problem.

\textsc{Sink-or-Clash} is the main focus of this paper. In order to study it, we define the minimum number of queries required to solve it.

\begin{definition}\label{def_t(n)}
    Let $t(n)$ be the minimum number of queries to deterministically solve \textsc{Sink-or-Clash} in an $n$-dimensional hypercube. Formally, given a deterministic algorithm $\mathcal{A}$ that solves \textsc{Sink-or-Clash} and an outmap $s \colon \mathfrak{C} \to \{0, 1\}^n$, let $t(\mathcal{A}, s)$ be the number of queries that algorithm $\mathcal{A}$ requires on the input outmap $s$. Then, \[t(n) = \min_\mathcal{A} \max_s t(\mathcal{A}, s).\]

    Similarly, define $\tilde{t}(n)$ to be the minimum number of queries a randomized algorithm needs in expectation to solve \textsc{Sink-or-Clash} in an $n$-dimensional hypercube. Analogously, given a randomized algorithm $\mathcal{A}$ that solves \textsc{Sink-or-Clash} and an outmap $s \colon \mathfrak{C} \to \{0, 1\}^n$, let $\tilde{t}(\mathcal{A}, s)$ be a random variable representing the number of queries that algorithm $\mathcal{A}$ requires on the input outmap $s$ in a particular instance. Then, \[\tilde{t}(n) = \min_\mathcal{A} \max_s \mathbb{E}(t(\mathcal{A}, s)).\] 
\end{definition}

We define the same notation for the promise problem Sink-USO.

\begin{definition}
    Similarly to \Cref{def_t(n)}, let $q(n)$ be the minimum number of queries to deterministically solve \textsc{Sink-USO} in an $n$-dimensional USO. Analogously, let $\tilde{q}(n)$ minimum number of queries a randomized algorithm needs in expectation to solve \textsc{Sink-USO} in an $n$-dimensional USO.
\end{definition}

The functions $q(n)$ and $\tilde{q}(n)$ have already been studied before. Let us summarize the known bounds.

The only values that are known exactly are $q(1) = 2$, $q(2) = 3$, $q(3) = 5$, $q(4) = 7$ \cite[Subsection 5.3]{schurrthesis} and $\tilde{q}(1) = 2$, $\tilde{q}(2) = 43/20$, $\tilde{q}(3) = 4074633/1369468$ \cite[Section 5]{upperuso}. Moreover, it is also known that $9 \leq q(5) \leq 12$ according to \cite[Proposition 5.8]{schurrthesis} and \cite[Theorem 3.3]{smallalg}.

In the deterministic setting, the best known lower bound was found by Schurr and Szabó by creating an adversary algorithm which ensures $q(n) \in \Omega\left(n^2 / \log(n)\right)$ \cite[Theorem 9]{quadlower}. The best known upper bound follows from the Improved Fibonacci Seesaw algorithm designed by Szabó and Welzl, which shows that $q(n) \in O\left(1.61^n\right)$ \cite[Theorem 4.1]{upperuso}.

In the randomized setting, the best upper bound comes from the Product algorithm, which states that $q(n) \leq q(k)q(n-k)$ and $\tilde{q}(n) \leq \tilde{q}(k)\tilde{q}(n-k)$ \cite[Lemma 3.2]{upperuso}. When using $\tilde{q}(3) = 4074633/1369468$, it is possible to obtain $\tilde{q}(n) \leq \tilde{q}(3)^{\lceil n/3 \rceil} \in O\left(1.438^n\right)$ by induction.

\subsection{Construction Techniques}

Matoušek showed that the number of $n$-dimensional USOs is $n^{\Theta(2^n)}$ \cite[Theorem 1]{usocount}. Due to this large number, it is infeasible to enumerate them in large dimensions. Hence, it is useful to study some techniques which allow the construction of USOs from other USOs. We will briefly explain some already known constructions without proof, which will be important for the results presented in the rest of the paper, especially in \Cref{cert-chapter}.

\begin{lemma}[Inherited Orientation]\label{inheorie}\cite[Lemma 3.1]{upperuso}
    Let $s$ be the outmap of an $n$-dimensional USO in $\mathfrak{C}$. Consider some subset $I \subseteq [n]$ of the dimensions. Let $\mathfrak{C}'$ be the $|I|$-dimensional hypercube and define the outmap $s' \colon \mathfrak{C}' \to \{0, 1\}^{|I|}$ as follows. Given $u \in \mathfrak{C}$, let $u'$ be the sink of $\mathfrak{C}^{[n]\backslash I}_u$ and define $s'(u_I) = s(u')_I$. Then, $s'$ is a USO.
\end{lemma}

This construction allows the construction of smaller USOs from larger ones, by looking only at the sinks of some faces. It will be essential for the Product Algorithm stated in \Cref{prodalg}.

\begin{figure}[htbp]
    \centering
    \begin{subfigure}[t]{0.45\textwidth}
        \centering
        \begin{tikzpicture}
            \draw[middlearrow={<}] (0,0,0) -- (2,0,0);
            \draw[middlearrow={<}] (2,0,0) -- (2,2,0);
            \draw[middlearrow={<}] (2,2,0) -- (0,2,0);
            \draw[middlearrow={>}] (0,2,0) -- (0,0,0);
            \draw[middlearrow={<}] (0,0,2) -- (2,0,2);
            \draw[middlearrow={<}] (2,0,2) -- (2,2,2);
            \draw[middlearrow={>}] (2,2,2) -- (0,2,2);
            \draw[middlearrow={<}] (0,2,2) -- (0,0,2);
            \draw[blue,middlearrow={<}] (0,0,0) -- (0,0,2);
            \draw[blue,middlearrow={>}] (2,0,0) -- (2,0,2);
            \draw[blue,middlearrow={<}] (0,2,0) -- (0,2,2);
            \draw[blue,middlearrow={<}] (2,2,0) -- (2,2,2);
            \fill [blue] (0, 2, 0) circle (2pt);
            \fill [blue] (0, 0, 0) circle (2pt);
            \fill [blue] (2, 2, 0) circle (2pt);
            \fill [blue] (2, 0, 2) circle (2pt);
        \end{tikzpicture}
        \caption{A USO}
        \label{uso_example2}
    \end{subfigure}
    \hfill
    \begin{subfigure}[t]{0.45\textwidth}
        \centering
        \begin{tikzpicture}
            \draw[middlearrow={<}] (0,0) -- (2,0);
            \draw[middlearrow={<}] (2,0) -- (2,2);
            \draw[middlearrow={<}] (2,2) -- (0,2);
            \draw[middlearrow={>}] (0,2) -- (0,0);
        \end{tikzpicture}
        \caption{An inherited orientation from \Cref{uso_example2}}
        \label{inherited}
    \end{subfigure}
    \caption{Inherited Orientation}
\end{figure}
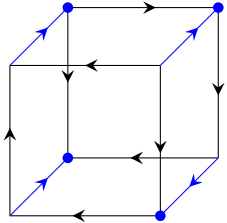
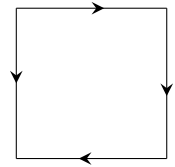

\begin{lemma}[Product USO]\label{product}\cite[Lemma 4.20]{schurrthesis}
    Let $s$ be the outmap of a $k$-dimensional USO in the hypercube $\mathfrak{C}_k$. For each $u \in \mathfrak{C}_k$, let $s_u$ be the outmap of a $(n-k)$ dimensional USO. Let $\mathfrak{C}$ be the $n$-dimensional hypercube. Then, the outmap $\tilde{s} \colon \mathfrak{C} \to \{0, 1\}^n$ defined by $\tilde{s}(u \cdot v) = s(u) \cdot s_u(v)$ for $u \in \mathfrak{C}_k, v \in \mathfrak{C}_{n-k}$ is a USO in $\mathfrak{C}$.
\end{lemma}

This construction allows us to construct larger USOs from smaller ones.

In the case $k = 1$, all $1$-edges are oriented in the same direction, and both facets can be any USO independently from each other. Given an USO, if there exists a dimension $i$ in which all $i$-edges are oriented in the same direction, the USO is called \emph{combed} \cite[Definition 4.27]{schurrthesis}.

In the case $k = n-1$, we have two copies of the same USO in the antipodal facets $\mathfrak{C}_{0^n}^{[n-1]}$ and $\mathfrak{C}_{1^n}^{[n-1]}$ and all $n$-edges can be oriented in any direction independently from each other.

\begin{figure}[htbp]
    \centering
    \begin{minipage}[t]{0.45\textwidth}
        \centering
        \begin{tikzpicture}
            \draw[middlearrow={<}] (0,0,0) -- (2,0,0);
            \draw[blue, middlearrow={<}] (2,0,0) -- (2,2,0);
            \draw[middlearrow={>}] (2,2,0) -- (0,2,0);
            \draw[blue, middlearrow={>}] (0,2,0) -- (0,0,0);
            \draw[middlearrow={>}] (0,0,2) -- (2,0,2);
            \draw[blue, middlearrow={<}] (2,0,2) -- (2,2,2);
            \draw[middlearrow={>}] (2,2,2) -- (0,2,2);
            \draw[blue, middlearrow={>}] (0,2,2) -- (0,0,2);
            \draw[middlearrow={<}] (0,0,0) -- (0,0,2);
            \draw[middlearrow={<}] (2,0,0) -- (2,0,2);
            \draw[middlearrow={>}] (0,2,0) -- (0,2,2);
            \draw[middlearrow={<}] (2,2,0) -- (2,2,2);
        \end{tikzpicture}
        \caption{\label{combed}A combed USO in the blue edges.}
    \end{minipage}
    \hfill
    \begin{minipage}[t]{0.45\textwidth}
        \centering
        \begin{tikzpicture}
            \coordinate (offset) at (\offsetx, \offsety, \offsetz);
            
            \draw[blue,middlearrow={<}] (0,0,0) -- (2,0,0);
            \draw[blue,middlearrow={<}] (2,0,0) -- (2,2,0);
            \draw[blue,middlearrow={<}] (2,2,0) -- (0,2,0);
            \draw[blue,middlearrow={>}] (0,2,0) -- (0,0,0);
            \draw[blue,middlearrow={<}] (0,0,2) -- (2,0,2);
            \draw[blue,middlearrow={<}] (2,0,2) -- (2,2,2);
            \draw[blue,middlearrow={<}] (2,2,2) -- (0,2,2);
            \draw[blue,middlearrow={>}] (0,2,2) -- (0,0,2);
            \draw[middlearrow={<}] (0,0,0) -- (0,0,2);
            \draw[middlearrow={>}] (2,0,0) -- (2,0,2);
            \draw[middlearrow={<}] (0,2,0) -- (0,2,2);
            \draw[middlearrow={<}] (2,2,0) -- (2,2,2);

            \draw[blue,middlearrow={<}] ($(0,0,0)+(offset)$) -- ($(2,0,0) + (offset)$);
            \draw[blue,middlearrow={<}] ($(2,0,0)+(offset)$) -- ($(2,2,0)+(offset)$);
            \draw[blue,middlearrow={<}] ($(2,2,0)+(offset)$) -- ($(0,2,0)+(offset)$);
            \draw[blue,middlearrow={>}] ($(0,2,0)+(offset)$) -- ($(0,0,0)+(offset)$);
            \draw[blue,middlearrow={<}] ($(0,0,2)+(offset)$) -- ($(2,0,2)+(offset)$);
            \draw[blue,middlearrow={<}] ($(2,0,2)+(offset)$) -- ($(2,2,2)+(offset)$);
            \draw[blue,middlearrow={<}] ($(2,2,2)+(offset)$) -- ($(0,2,2)+(offset)$);
            \draw[blue,middlearrow={>}] ($(0,2,2)+(offset)$) -- ($(0,0,2)+(offset)$);
            \draw[middlearrow={<}] ($(0,0,0)+(offset)$) -- ($(0,0,2)+(offset)$);
            \draw[middlearrow={>}] ($(2,0,0)+(offset)$) -- ($(2,0,2)+(offset)$);
            \draw[middlearrow={<}] ($(0,2,0)+(offset)$) -- ($(0,2,2)+(offset)$);
            \draw[middlearrow={<}] ($(2,2,0)+(offset)$) -- ($(2,2,2)+(offset)$);

            \draw[middlearrow={<}] (0,0,0) -- ($(0,0,0)+(offset)$);
            \draw[middlearrow={<}] (2,0,0) -- ($(2,0,0)+(offset)$);
            \draw[middlearrow={<}] (2,2,0) -- ($(2,2,0)+(offset)$);
            \draw[middlearrow={>}] (0,2,0) -- ($(0,2,0)+(offset)$);
            \draw[middlearrow={<}] (0,0,2) -- ($(0,0,2)+(offset)$);
            \draw[middlearrow={<}] (2,0,2) -- ($(2,0,2)+(offset)$);
            \draw[middlearrow={>}] (2,2,2) -- ($(2,2,2)+(offset)$);
            \draw[middlearrow={<}] (0,2,2) -- ($(0,2,2)+(offset)$);
        \end{tikzpicture}
        \caption{A product of two $2$-dimensional USO}
        \label{productimage}
    \end{minipage}
\end{figure}

\begin{definition}[Hypervertex and Hypersink]\label{hypervertex}\cite[Definition 4.17]{schurrthesis}
    Let $s \colon \mathfrak{C} \to \{0, 1\}^n$ be the outmap of a USO. Given some $u \in \mathfrak{C}$ and $I \subseteq [n]$, the face $\mathfrak{C}_u^I$ is called a \emph{hypervertex} if for all $v, w \in \mathfrak{C}_u^I$, $s(v)_{[n] \backslash I} = s(w)_{[n] \backslash I}$. The face is called a \emph{hypersink} if for all $v \in \mathfrak{C}_u^I$, $s(v)_{[n]\backslash I} = 0^{n - |I|}$.
\end{definition}

\begin{lemma}[Hypervertex Replacement]\label{flip}\cite[Corollary 4.19]{schurrthesis}
    Let $s \colon \mathfrak{C} \to \{0, 1\}^n$ be the outmap of a USO. Let $\mathfrak{C}_u^I$ be a hypervertex, for some $u \in \mathfrak{C}$ and $I \subseteq [n]$. Let $s'$ be the outmap of any USO in $\mathfrak{C}_u^I$. Define the outmap $\tilde{s} \colon \mathfrak{C} \to \{0, 1\}^n$ as $\tilde{s}(v) = s(v)$ for $u \notin \mathfrak{C}_u^I$ and $\tilde{s}(v)_I = s'(v)$, $\tilde{s}(v)_{[n] \backslash I} = s(v)_{[n] \backslash I}$ for $v \in \mathfrak{C}_u^I$. Then, $\tilde{s}$ is a USO.
\end{lemma}

The hypervertex replacement is a technique which allows the replacement of the orientations inside a hypervertex by any other USO, while maintaining the USO property. It has some special cases.

When $|I| = 1$, the hypervertex consists of two vertices $u$ and $u \oplus e_i$ such that $s(u)_{[n] \backslash \{i\}} = s(u \oplus e_i)_{[n] \backslash\{i\}}$. In this case, it is possible to set the orientation of the edge $(u, u\oplus e_i)$ in any direction. This is called a \emph{flippable edge}.

When $|I| = n-1$, there exists some hypervertex $\mathfrak{C}_u^{[n] \backslash \{i\}}$ and so all $i$-edges are oriented in the same direction. \Cref{flip} states that we can replace $\mathfrak{C}_u^{[n] \backslash \{i\}}$ with any other USO. This was already seen in \Cref{product}, as this USO would be combed in the $i$th dimension.

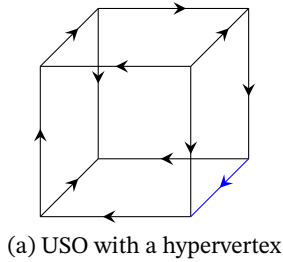
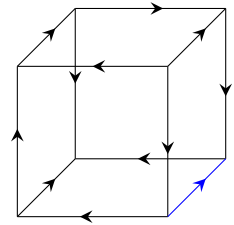
\begin{figure}[htbp]
    \centering
    \begin{subfigure}[t]{0.45\textwidth}
        \centering
        \begin{tikzpicture}
            \draw[middlearrow={<}] (0,0,0) -- (2,0,0);
            \draw[middlearrow={<}] (2,0,0) -- (2,2,0);
            \draw[middlearrow={<}] (2,2,0) -- (0,2,0);
            \draw[middlearrow={>}] (0,2,0) -- (0,0,0);
            \draw[middlearrow={<}] (0,0,2) -- (2,0,2);
            \draw[middlearrow={<}] (2,0,2) -- (2,2,2);
            \draw[middlearrow={>}] (2,2,2) -- (0,2,2);
            \draw[middlearrow={<}] (0,2,2) -- (0,0,2);
            \draw[middlearrow={<}] (0,0,0) -- (0,0,2);
            \draw[blue,middlearrow={>}] (2,0,0) -- (2,0,2);
            \draw[middlearrow={<}] (0,2,0) -- (0,2,2);
            \draw[middlearrow={<}] (2,2,0) -- (2,2,2);
        \end{tikzpicture}
        \caption{USO with a hypervertex}
        \label{orig_uso_flip_edge}
    \end{subfigure}
    \hfill
    \begin{subfigure}[t]{0.45\textwidth}
        \centering
        \begin{tikzpicture}
            \draw[middlearrow={<}] (0,0,0) -- (2,0,0);
            \draw[middlearrow={<}] (2,0,0) -- (2,2,0);
            \draw[middlearrow={<}] (2,2,0) -- (0,2,0);
            \draw[middlearrow={>}] (0,2,0) -- (0,0,0);
            \draw[middlearrow={<}] (0,0,2) -- (2,0,2);
            \draw[middlearrow={<}] (2,0,2) -- (2,2,2);
            \draw[middlearrow={>}] (2,2,2) -- (0,2,2);
            \draw[middlearrow={<}] (0,2,2) -- (0,0,2);
            \draw[middlearrow={<}] (0,0,0) -- (0,0,2);
            \draw[blue,middlearrow={<}] (2,0,0) -- (2,0,2);
            \draw[middlearrow={<}] (0,2,0) -- (0,2,2);
            \draw[middlearrow={<}] (2,2,0) -- (2,2,2);
        \end{tikzpicture}
        \caption{USO obtained from \Cref{orig_uso_flip_edge} by changing the hypervertex}
    \end{subfigure}
    \caption{USOs with a flippable edge, drawn in blue.}
    \label{flippable image}
\end{figure}

\begin{lemma}[Partial Swap]\label{partswap}\cite[Theorem 18]{univconstr}
    Let $s \colon \mathfrak{C} \to \{0, 1\}^n$ be the outmap of a USO. Let $i \in [n]$ be any dimension. Let $\tilde{s} \colon \mathfrak{C} \to \{0, 1\}^n$ be the outmap such that \[\tilde{s}(u) = \begin{cases}
        s(u) & \text{if } u_i = s(u)_i \\
        s(u \oplus e_i) & \text{if } u_i \neq s(u)_i 
    \end{cases}.\] Then, $\tilde{s}$ is a USO.
\end{lemma}

The partial swap consists in selecting one dimension $i$, and swapping the outmaps of $u$ and $u \oplus e_i$ if and only if the $i$-edge between them is oriented towards $\mathfrak{C}_{1^n}^{[n] \backslash \{i\}}$.

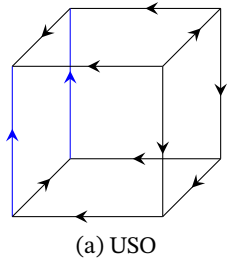
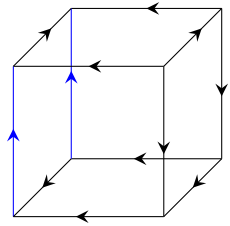
\begin{figure}[htbp]
    \centering
    \begin{subfigure}[t]{0.45\textwidth}
        \centering
        \begin{tikzpicture}
            \draw[middlearrow={<}] (0,0,0) -- (2,0,0);
            \draw[middlearrow={<}] (2,0,0) -- (2,2,0);
            \draw[middlearrow={>}] (2,2,0) -- (0,2,0);
            \draw[blue,middlearrow={<}] (0,2,0) -- (0,0,0);
            \draw[middlearrow={<}] (0,0,2) -- (2,0,2);
            \draw[middlearrow={<}] (2,0,2) -- (2,2,2);
            \draw[middlearrow={>}] (2,2,2) -- (0,2,2);
            \draw[blue,middlearrow={<}] (0,2,2) -- (0,0,2);
            \draw[middlearrow={<}] (0,0,0) -- (0,0,2);
            \draw[middlearrow={>}] (2,0,0) -- (2,0,2);
            \draw[middlearrow={>}] (0,2,0) -- (0,2,2);
            \draw[middlearrow={<}] (2,2,0) -- (2,2,2);
        \end{tikzpicture}
        \caption{USO}
        \label{origusoimage}
    \end{subfigure}
    \hfill
    \begin{subfigure}[t]{0.45\textwidth}
        \centering
        \begin{tikzpicture}
            \draw[middlearrow={<}] (0,0,0) -- (2,0,0);
            \draw[middlearrow={<}] (2,0,0) -- (2,2,0);
            \draw[middlearrow={>}] (2,2,0) -- (0,2,0);
            \draw[blue,middlearrow={<}] (0,2,0) -- (0,0,0);
            \draw[middlearrow={<}] (0,0,2) -- (2,0,2);
            \draw[middlearrow={<}] (2,0,2) -- (2,2,2);
            \draw[middlearrow={>}] (2,2,2) -- (0,2,2);
            \draw[blue,middlearrow={<}] (0,2,2) -- (0,0,2);
            \draw[middlearrow={>}] (0,0,0) -- (0,0,2);
            \draw[middlearrow={>}] (2,0,0) -- (2,0,2);
            \draw[middlearrow={<}] (0,2,0) -- (0,2,2);
            \draw[middlearrow={<}] (2,2,0) -- (2,2,2);
        \end{tikzpicture}
        \caption{USO obtained from \Cref{origusoimage} by applying a partial swap}
        \label{partswapimage}
    \end{subfigure}
    \caption{Partial Swap}
\end{figure}

\section{Lower Bound and Algorithms for \textsc{Sink-or-Clash}}\label{non-promise}

In this chapter, we establish some bounds in the optimal number of queries needed to solve \textsc{Sink-or-Clash}, initially in the deterministic setting and then by allowing randomized algorithms.

\begin{lemma}\label{notharder}
    \textsc{Sink-or-Clash} is not easier \textsc{Sink-USO}, i.e., $q(n) \leq t(n)$ and $\tilde{q}(n) \leq \tilde{t}(n)$.
\end{lemma}
\begin{proof}
    Let $\mathcal{A}$ be an optimal algorithm for \textsc{Sink-or-Clash}. When the input is a USO, $\mathcal{A}$ cannot find a clash (as none exist), so it must return a sink. Hence, $\mathcal{A}$ also solves \textsc{Sink-USO}, using at most $t(n)$ queries. Since $q(n)$ minimizes over a superset of algorithms as compared to $t(n)$, we have $q(n) \leq t(n)$.
\end{proof}

This shows that any known lower bound for \textsc{Sink-USO} is also a lower bound for \textsc{Sink-or-Clash}, both in the deterministic and randomized settings. Thus, the lower bound of Schurr and Szabó \cite[Theorem 9]{quadlower} carries over.

\begin{corollary}
    $t(n) \in \Omega\left(n^2/\log(n)\right)$.
\end{corollary}

\subsection{Upper Bounds}

In this subsection, we adapt some algorithms from \textsc{Sink-USO} to \textsc{Sink-or-Clash}, which create the same upper bounds on the number of queries both in \textsc{Sink-USO} and in \textsc{Sink-or-Clash}.

Initially, we state how the Product Algorithm \cite[Lemma 3.2]{upperuso} can be adapted to also solve \textsc{Sink-or-Clash}, as shown below.

\begin{lemma}[Adapted Product Algorithm]\label{prodalg}
    $t(n) \leq t(k)t(n-k)$ and $\tilde{t}(n) \leq \tilde{t}(k) \tilde{t}(n-k)$.
\end{lemma}
\begin{proof}
    This proof is similar to the one Szabó and Welzl presented in \cite[Lemma 3.2]{upperuso}, with a slight modification to include the case when the outmap is not a USO.

    Let $s \colon \mathfrak{C} \to \{0, 1\}^n$ be an outmap in the $n$-dimensional $\mathfrak{C}$. Let $I \coloneq [k] \subseteq [n]$. Create the inherited outmap $s'$ as given by \Cref{inheorie} in the $k$-dimensional hypercube, denoted as $\mathfrak{C}'$.

    We perform this same algorithm in the smaller hypercube $\mathfrak{C}'$ to solve $s'$. Now, the vertex evaluations are implemented as solving an $(n-k)$-dimensional USO, using the same algorithm. When solving a vertex evaluation, if a clash is found, it is also a clash in $s$. Otherwise, a sink is found, which is then used to find the outmap of a vertex in the inherited orientation $s'$.

    After $t(k)t(n-k)$ queries (or after $\tilde{t}(k)\tilde{t}(n-k)$ expected queries in the randomized setting), either a sink or a clash if found in $s'$. If a sink is found, that must be a global sink in $s$. Otherwise, suppose that the algorithm found vertices $u', v' \in \mathfrak{C}'$ which clash. Suppose they correspond to the sinks $u, v \in \mathfrak{C}$. Therefore, there exists no $i \in [k]$ such that $s(u)_i \neq s(v)_i \wedge u_i \neq v_i$. Moreover, $s(u)_i = s(v)_i = 0$ for all $i \in [n] \backslash [k]$ as they are sinks in $(n-k)$-dimensional hypercubes, so it follows that there exists no $i \in [n]$ such that $s(u)_i \neq s(v)_i \wedge u_i \neq v_i$, i.e., $u$ and $v$ clash in $s$.
\end{proof}

By using induction and the fact that $t$ and $\tilde{t}$ are non-decreasing, it follows that $t(n) \leq t(k)^{\lceil n/k \rceil}$ and $\tilde{t}(n) \leq \tilde{t}(k)^{\lceil n/k \rceil}$ \cite[Corollary 3.3]{upperuso}. When the exact value of $t(k)$ or $\tilde{t}(k)$ are known for some small dimension $k$, then this creates the exponential upper bound $t(n) \in O\left(\sqrt[k]{t(k)}\right)$ and $\tilde{t}(n) \in O\left(\sqrt[k]{\tilde{t}(k)}\right)$. These bounds are especially important in the randomized setting, as the best known randomized algorithm for \textsc{Sink-USO} uses the fact that $\tilde{q}(3) = 4074633/1369468 \approx 2.976$ to create an upper bound $\tilde{q}(n) \in O\left(\sqrt[3]{2.976}^n\right) = O\left(1.438^n\right)$ \cite[Section 5]{upperuso}.

Moreover, an adaption Fibonacci Seesaw \cite[Section 4]{upperuso} also works for \textsc{Sink-or-Clash}. Such adaptation is shown as \Cref{fib}, which should be initially called with $J = [n]$ and any vertex $z \in \mathfrak{C}$.

\begin{lemma}
    $t(n) \in O\left(1.61^n\right)$.
\end{lemma}

\begin{algorithm}
\caption{Adapted\_Fibonacci\_Seesaw}
\begin{algorithmic}[1]\label{fib}
\STATE \textbf{Input:} An outmap $s$ of an $n$-dimensional hypercube $\mathfrak{C}$, a subset $J \subseteq [n]$ of the dimensions, a vertex $z$
\STATE \textbf{Output:} A sink or a clash in the face $\mathfrak{C}_z^J$
\IF{$|J| = 0$}
    \RETURN z
\ENDIF
\STATE $u \gets z$
\STATE $v \gets z \oplus \left(\bigvee_{i \in J} e_i\right)$
\STATE $I \gets \emptyset$
\FOR{$k \gets 0$ to $|J|-2$}
    \STATE $j \gets 0$
    \FOR{$i \in J \backslash I$}
        \IF{$\left(u_i \oplus v_i\right) \wedge \left(s(u)_i \oplus s(v)_i\right) = 1$}
            \STATE $j \gets i$
            \BREAK
        \ENDIF
    \ENDFOR
    \IF{$j = 0$}
        \RETURN $(u, v)$
    \ENDIF
    \IF{$s(u)_j = 1$}
        \STATE \textbf{swap} $(u, v)$
    \ENDIF
    \STATE $w \gets$ Adapted\_Fibonacci\_Seesaw$\left(s, I, v \oplus e_j\right)$
    \IF{$w$ is a clash}
        \RETURN $w$
    \ENDIF
    \IF{$s(w)_j = 1$}
        \RETURN $(v, w)$
    \ENDIF
    \STATE $v \gets w$
    \STATE $I \gets I \cup \{j\}$
\ENDFOR
\STATE $j \leftarrow J \backslash I$
\IF{$s(u)_j = 0$}
    \RETURN $u$
\ENDIF
\IF{$s(v)_j = 0$}
    \RETURN $v$
\ENDIF
\RETURN $(u, v)$
\end{algorithmic}
\end{algorithm}

\begin{proof}

Given an $n$-dimensional hypercube $\mathfrak{C}$, the algorithm maintains sinks of a pair of antipodal $k$-dimensional faces from $k=0$ until $k=n-1$. Initially, $2$ evaluations are used to evaluate two antipodal vertices, which are $0$-dimensional faces.

For some $0 \leq k \leq n-2$, suppose the algorithm has two $k$-dimensional antipodal faces $\mathfrak{C}_u^I$ and $\mathfrak{C}_v^I$ for some $I \subseteq [n]$, with sinks $u$ and $v$, respectively. If $u$ and $v$ clash, the algorithm ends. Otherwise, there exists some $j \in [n]$ such that $s(u)_j \neq s(v)_j \wedge u_j \neq v_j$. As $u$ and $v$ are sinks in the respective faces, then $s(u)_I = s(v)_I = 0^k$, so $j \notin I$. For that specific $j$, suppose without loss of generality that $s(u)_j = 0 \neq s(v)_j = 1$. Afterwards, consider the antipodal faces $\mathfrak{C}_u^{I \cup \{j\}}$ and $\mathfrak{C}_v^{I \cup \{j\}}$. As $s(u)_{I \cup \{j\}} = 0^{k+1}$, then $u$ is a sink in $\mathfrak{C}_u^{I \cup \{j\}}$. Moreover, $v$ is a sink of a facet of $\mathfrak{C}_v^{I \cup \{j\}}$, and as $s(v)_j = 1$, it implies that if a sink exists in $\mathfrak{C}_{v \oplus e_j}^I$, it must be a sink in $\mathfrak{C}_v^{I \cup \{j\}}$ or it would clash with $v$. Furthermore, $t(k)$ queries are enough to find the sink or a clash in $\mathfrak{C}_{v \oplus e_j}^I$, which is also a clash in $\mathfrak{C}$.

Finally, when the algorithm found the sink of two antipodal facets, in order for them not to clash, one of them must be the sink of $\mathfrak{C}$. Hence, the algorithm is correct.

The algorithm performs in total $2 + \sum_{k=0}^{n-2} t(k)$ queries, so $t(n) \leq 2 + \sum_{k=0}^{n-2} t(k)$ for $n \geq 1$, which solves to $t(n) \in O\left(1.62^n\right)$.

Later, we show that $t(4) = 7$ in \Cref{t4=7}. Therefore, an adaptation of the Improved Fibonacci Seesaw, as presented by Szabó and Welzl, also works \cite[Theorem 4.1]{upperuso}. In this algorithm, when two antipodal $4$-dimensional faces have been evaluated, the algorithm uses the product algorithm and evaluates at most another five $(n-4)$-dimensional faces. This creates the recursion \[t(n) \leq 2 + \sum_{k=0}^{n-5} t(k) + 5t(n-4)\] which solves to $t(n) \in O\left(1.61^n\right)$.

\end{proof}

\subsection{Algorithms for Small Dimensions}

Many of the known algorithms still work when the outmap does not necessarily encode a USO. This is true for the optimal algorithms when $n \leq 3$, where the same number of queries are enough.

\begin{lemma}\label{t1}
    $t(1) = 2$
\end{lemma}
\begin{proof}
    One query is not enough as an adversary algorithm could state that the first vertex is the source. Moreover, two queries are enough as they cover all the vertices.
\end{proof}

\begin{lemma}\label{t2}
    $t(2) = 3$
\end{lemma}
\begin{proof}
    As $q(2) = 3$ \cite[Proposition 5.5]{schurrthesis}, then $t(2) \geq 3$. Moreover, the Adapted Fibonacci Seesaw in \Cref{fib} states that $t(2) \leq 2 + t(0) = 3$.
\end{proof}

\begin{lemma}\label{t3}
    $t(3) = 5$
\end{lemma}
\begin{proof}
    Similarly, as $q(3) = 5$ \cite[Proposition 5.6]{schurrthesis}, then $t(3) \geq 5$. Moreover, the Adapted Fibonacci Seesaw in \Cref{fib} states that $t(3) \leq 2 + t(0) + t(1) = 5$.
\end{proof}

\Cref{t2} and \Cref{t3} are not too surprising, as the optimal algorithm for these small dimensions is the Fibonacci Seesaw, which we have already shown to be adaptable to \textsc{Sink-or-Clash}.

Unfortunately, the SevenStepsToHeaven Algorithm \cite[Lemma 6.1]{upperuso} does not solve \textsc{Sink-or-Clash}. The main problem is that it assumes the orientation of some edges due to the fact that the outmap corresponds to a USO, but when the outmap does not need to be a USO, disrespecting that assumption does not immediately create a clash. Therefore, it cannot be used to prove that $q(4) \leq 7$.

\begin{lemma}
    The SevenStepsToHeaven Algorithm does not solve \textsc{Sink-or-Clash}.
\end{lemma}

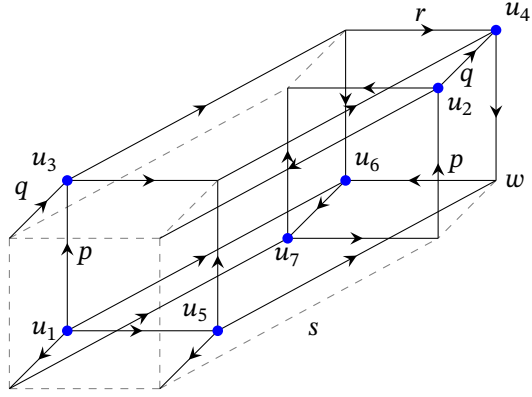
\begin{figure}[htbp]\centering
\begin{tikzpicture}
        \coordinate (offset) at (\offsetx, \offsety, \offsetz);
            
        \draw[middlearrow={>}] (0,0,0) -- (2,0,0);
        \draw[middlearrow={>}] (2,0,0) -- (2,2,0);
        \draw[middlearrow={<}] (2,2,0) -- (0,2,0);
        \draw[middlearrow={<}] (0,2,0) -- (0,0,0);
        \draw[dashed, opacity=0.5] (0,0,2) -- (2,0,2);
        \draw[dashed, opacity=0.5] (2,0,2) -- (2,2,2);
        \draw[dashed, opacity=0.5] (2,2,2) -- (0,2,2);
        \draw[dashed, opacity=0.5] (0,2,2) -- (0,0,2);
        \draw[middlearrow={>}] (0,0,0) -- (0,0,2);
        \draw[middlearrow={>}] (2,0,0) -- (2,0,2);
        \draw[middlearrow={<}] (0,2,0) -- (0,2,2);
        \draw[dashed, opacity=0.5] (2,2,0) -- (2,2,2);

        \draw[middlearrow={<}] ($(0,0,0)+(offset)$) -- ($(2,0,0) + (offset)$);
        \draw[middlearrow={<}] ($(2,0,0)+(offset)$) -- ($(2,2,0)+(offset)$);
        \draw[middlearrow={<}] ($(2,2,0)+(offset)$) -- ($(0,2,0)+(offset)$);
        \draw[middlearrow={>}] ($(0,2,0)+(offset)$) -- ($(0,0,0)+(offset)$);
        \draw[middlearrow={>}] ($(0,0,2)+(offset)$) -- ($(2,0,2)+(offset)$);
        \draw[middlearrow={>}] ($(2,0,2)+(offset)$) -- ($(2,2,2)+(offset)$);
        \draw[middlearrow={>}] ($(2,2,2)+(offset)$) -- ($(0,2,2)+(offset)$);
        \draw[middlearrow={<}] ($(0,2,2)+(offset)$) -- ($(0,0,2)+(offset)$);
        \draw[middlearrow={>}] ($(0,0,0)+(offset)$) -- ($(0,0,2)+(offset)$);
        \draw[dashed, opacity=0.5] ($(2,0,0)+(offset)$) -- ($(2,0,2)+(offset)$);
        \draw[dashed, opacity=0.5] ($(0,2,0)+(offset)$) -- ($(0,2,2)+(offset)$);
        \draw[middlearrow={<}] ($(2,2,0)+(offset)$) -- ($(2,2,2)+(offset)$);

        \draw[middlearrow={>}] (0,0,0) -- ($(0,0,0)+(offset)$);
        \draw[middlearrow={>}] (2,0,0) -- ($(2,0,0)+(offset)$);
        \draw[middlearrow={>}] (2,2,0) -- ($(2,2,0)+(offset)$);
        \draw[middlearrow={>}] (0,2,0) -- ($(0,2,0)+(offset)$);
        \draw[middlearrow={>}] (0,0,2) -- ($(0,0,2)+(offset)$);
        \draw[dashed, opacity=0.5] (2,0,2) -- ($(2,0,2)+(offset)$);
        \draw[middlearrow={<}] (2,2,2) -- ($(2,2,2)+(offset)$);
        \draw[dashed, opacity=0.5] (0,2,2) -- ($(0,2,2)+(offset)$);
            
        \fill[blue] (0,0,0) circle (2pt);
        \fill[blue] ($(2,2,2)+(offset)$) circle (2pt);
        \fill[blue] (0,2,0) circle (2pt);
        \fill[blue] ($(2,2,0)+(offset)$) circle (2pt);
        \fill[blue] (2,0,0) circle (2pt);
        \fill[blue] (offset) circle (2pt);
        \fill[blue] ($(0,0,2)+(offset)$) circle (2pt);

        \node[left] at (0,0,0) {$u_1$};
        \node[below right] at ($(2,2,2)+(offset)$) {$u_2$};
        \node[above left] at (0,2,0) {$u_3$};
        \node[above right] at ($(2,2,0)+(offset)$) {$u_4$};
        \node[right] at ($(2,0,0)+(offset)$) {$w$};
        \node[above left] at (2,0,0) {$u_5$};
        \node[above right] at (offset) {$u_6$};
        \node[below] at ($(0,0,2)+(offset)$) {$u_7$};

        \node[right] at (0,1,0) {$p$};
        \node[right] at ($(2,1,2)+(offset)$) {$p$};
        \node[above left] at (0,2,1) {$q$};
        \node[below] at ($(2,2,1)+(offset)$) {$q$};
        \node[above] at ($(1,2,0)+(offset)$) {$r$};
        \node[below right] at ($(2,0,2)+0.5*(offset)$) {$s$};
\end{tikzpicture}
    \caption{Failed SevenStepsToHeaven Algorithm.}
    \label{failseven}
\end{figure}
\begin{proof}
    For simplicity, we will use the same notation as Szabó and Welzl. We will also follow the algorithm (presented in \cite[Lemma 6.1]{upperuso}) until the end, when it fails to find a sink or a clash if we assume that the outmap is not a USO. The queried vertices are represented in \Cref{failseven}.

    Initially, two antipodal vertices $u_1$ and $u_2$ are evaluated. If they clash, the algorithm ends. Otherwise, there must exist some dimension $p$ such that $s(u_1)_p = 1 \neq s(u_2)_p = 0$.

    If $u_2$ has another incoming edge, then the algorithm finishes in $1+3+3=7$ steps using the product algorithm $t(4) \leq t(2)t(2)$, as the sink of one of the squares was found in $1$ instead of $3$ steps. So, assume otherwise.

    Afterwards, the algorithm evaluates $u_3 \coloneq u_1 \oplus e_p$. If $s(u_3)_p = 1$, then it would clash with $u_1$, so assume $s(u_3)_p = 0$. In order for $u_3$ and $u_2$ not to clash, there must exist some dimension $q$ in which $s(u_3)_q = 0$. If $u_3$ has another incoming edge, it would be the sink of a $3$-cube, so the algorithm would finish in $2+5=7$ steps using the product algorithm $t(4) \leq t(3)t(1)$, as the sink of the $3$-cube was found in $2$ instead of $5$ steps. So, assume otherwise.

    After that, evaluate $u_4 \coloneq u_2 \oplus e_q$. In order for $u_4$ and $u_2$ not to clash, then $s(u_4)_q = 0$. If $s(u_4)_p = 0$, then $u_3$ and $u_4$ are the sinks of two antipodal $2$-cubes, so $3$ more queries are enough using the product algorithm, by finding the sink in the appropriate $2$-cube. So, assume $s(u_4)_p = 1$. Finally, in order for $u_3$ and $u_4$ not to clash, $u_4$ must have another incoming edge, say $s(u_4)_r = 0$, and call the other dimension $s$.

    Now, the SevenStepsToHeaven assumes that $w \coloneq u_4 \oplus e_p$ is the sink of $\mathfrak{C}_{u_2}^{\{p, q\}}$. While this is true if it is a USO, this assumption fails as it is possible to contradict it without creating a clash, as seen below.

    Then, evaluate $u_5 \coloneq w \oplus e_s$. Assume this is the case when $s(u_5)_s = 1$.

    The algorithm would evaluate $u_6 \coloneq w \oplus e_r$ next. If $s(u_6)_r = 1$, then $w$ would be forced to either be a sink or to clash with another vertex. Hence, assume $s(u_6)_r = 0$. In order for $u_6$ and $u_4$ not to clash, then $s(u_6)_p = 0$. In order for $u_6$ and $u_3$ not to clash, then $s(u_6)_s = 0$. If $s(u_6)_q = 0$, then $u_6$ would be the sink. So, assume that $s(u_6)_q = 1$.

    If this was a USO, then $u_7 \coloneq u_6 \oplus e_q$ would need to be a sink. However, it is possible to define $s(u_7)_p = 1$, $s(u_7)_q = 0$, $s(u_7)_r = 1$ and $s(u_7)_s = 0$, and no pair of the queried vertices clashes.

    Hence, the SevenStepsToHeaven algorithm does not solve \textsc{Sink-or-Clash}. \Cref{failseven} shows the orientation following this algorithm, with $7$ evaluated vertices, in which no pair clashes.
\end{proof}

Although the SevenStepsToHeaven algorithm does not solve \textsc{Sink-or-Clash}, $7$ queries are still enough.

\begin{lemma}\label{t4=7}
    $t(4) = 7$.
\end{lemma}

This value was found computationally, using \Cref{choosevertex} and \Cref{chooseoutmap}.

\begin{figure}[htbp]\centering
\begin{algorithm}[H]
\caption{Choose\_Vertex}
\begin{algorithmic}[1]\label{choosevertex}
\STATE \textbf{Input:} A dimension $n$, a hypercube $\mathfrak{C}$, the number of queries $q$, a list of known vertices $T$, a map from the known vertices to their outmaps $s$ where no pair of vertices clash and none is the sink
\STATE \textbf{Output:} A boolean which is True if and only if a vertex can be picked which ensures a clash or a sink is found using at most $q$ queries, including the ones already made in $T$
\IF{$\textbf{length}(T) = q$}
    \RETURN \FALSE
\ENDIF
\FOR{each $u$ in $\mathfrak{C} \backslash T$}
    \IF{$\textbf{not} \text{ Choose\_Outmap}(n, \mathfrak{C}, q, T, s, u)$}
        \RETURN \TRUE
    \ENDIF
\ENDFOR
\RETURN \FALSE
\end{algorithmic}
\end{algorithm}

\begin{algorithm}[H]
\caption{Choose\_Outmap}
\begin{algorithmic}[1]\label{chooseoutmap}
\STATE \textbf{Input:} A dimension $n$, a hypercube $\mathfrak{C}$, the number of queries $q$, a list of known vertices $T$, a map from the known vertices to their outmaps $s$ where no pair of vertices clash and none is the sink, an unassigned vertex $u$
\STATE \textbf{Output:} A boolean which is True if and only if an outmap can be picked for $u$ such that it is possible to ensure neither a sink not a clash is found using at most $q$ queries, including the ones already made in $T$
\FOR{$outmap \in \{0, 1\}^n \backslash \left\{0^n\right\}$}
    \STATE $fail \gets \FALSE$
    \FOR{$v \in T$}
        \IF{$(outmap \oplus s(v)) \wedge (u \oplus v) = 0^n$}
            \STATE $fail \gets \TRUE$
            \BREAK
        \ENDIF
    \ENDFOR
    \IF{\textbf{not} fail \AND $\textbf{not} \text{ Choose\_Vertex}(n, \mathfrak{C}, q, T \cup \{u\}, s \cup \langle u, outmap \rangle)$}
        \RETURN \TRUE
    \ENDIF
\ENDFOR
\RETURN \FALSE
\end{algorithmic}
\end{algorithm}
\end{figure}

These algorithms simulate a two-player game. The first player picks a vertex in $\mathfrak{C}$ and the second player returns an outmap. The second player loses if he cannot return any outmap which is not a sink or that does not clash with a previously known vertex. The first player wins if he can ensure that the second player loses within some determined number of queries.

In addition to that, it uses the initial symmetry of the problem to reduce the number of branches of the game tree it searches.

Initially, it can be assumed that the first vertex the first player queries $0^n$.

After that, there are only $n$ possible non-sink values the second player can give, up to symmetry, which correspond to how many $1$'s appear in the bitstring representation of the value. The algorithm considers as potentials values for $s(0^n)$ the bitstrings of the form $\bigvee_{l=1}^i e_l$ for any $i \in [n]$.

Finally, for each fixed $i$, there is a partition of the dimensions $[n]$ in two subsets $[i]$ and $[n] \backslash [i]$, corresponding to those that have an outgoing and incoming edge to $0^n$, respectively. Hence, the next vertex the first player chooses only depends on the number of $1$'s in the bitstring within both subsets independently.

If for some $i$, the corresponding outmap given to $0^n$ is such that for any vertex the first player picks next, the second player can always win, then the second player has a winning strategy of choosing $s(0^n) = \bigvee_{l=1}^i e_l$. Otherwise, the first player wins.

When analyzing the $5$-dimensional hypercube, according to \cite[Proposition 5.8]{schurrthesis} and \cite[Theorem 3.3]{smallalg}, it follows that $9 \leq q(5) \leq 12$. Therefore, $t(5) \geq 9$. Moreover, the Adapted Improved Fibonacci Seesaw yields $t(5) \leq 2 + t(0) + 5t(1) = 13$.

For $n \geq 5$, it is infeasible to use \Cref{choosevertex} and \Cref{chooseoutmap} to exactly compute $t(n)$, even when using the initial symmetry of the cube. Furthermore, we were still able to discover using \Cref{choosevertex} and \Cref{chooseoutmap} that the second player wins when $n = 5$ and at most $10$ queries are used, so it follows that $t(5) \geq 11$.

\begin{lemma}\label{t5}
    $11 \leq t(5) \leq 13$.
\end{lemma}

\subsection{Randomized Algorithms}

Although in up to dimension $4$ the number of required queries to solve \textsc{Sink-or-Clash} and \textsc{Sink-USO} deterministically is the same, this will not happen when randomized algorithms are allowed.

For a $1$-dimensional hypercube, when the chosen outmap $s$ is such that $s(0) = s(1) = 1$, then both vertices are sources. As no algorithm is able to find a sink, and it must evaluate both vertices to find the clash. Moreover, evaluating both vertices is always enough, as it receives the full information about the outmap. Therefore, $\tilde{t}(1) = 2$, which is already different from \textsc{Sink-USO}, as $\tilde{q}(1) = 1.5$.

\begin{lemma}
    $\tilde{t}(1) = 2$.
\end{lemma}

For a $2$-dimensional hypercube, there are already $2^{2 \cdot 2^2} = 256$ possible outmaps (some of them isomorphic), so this already becomes infeasible to compute by hand.

Tessaro designed a linear program whose optimal solution corresponds to the expected number of queries in the optimal randomized algorithm for \textsc{Sink-USO} \cite[Equation 4.2]{3rand}. In his thesis, he considers the equivalent two player game, in which initially the adversary picks an outmap and then the algorithm (which we denote by the main player) performs some queries until a sink is found. By considering the payoff of a branch of the game tree to be the number of queries made, the main player wishes to minimize that number, using some randomized strategy.

He defines a history $H$ to be a sequence of vertices and their outmaps, such that there is no sink or clash. This is all the information that the main player has at a certain state of the game. Afterwards, he lets a sequence $S$ be a history together with a vertex that has not been evaluated yet. Finally, he lets $\Omega_n$ to be the set of all $n$-dimensional USOs, $\mathcal{H}_n$ to be the set of all histories and $\mathcal{S}_n$ to be the set of all sequences.

So, the strategy of the adversary can be encoded as a vector $y \in \mathbb{R}^{|\Omega_n|}$, where $y_\mathcal{O}$ is the probability of the adversary picking the outmap $\mathcal{O} \in \Omega_n$. The strategy of the main player can be explained as a vector $x \in \mathbb{R}^{|\mathcal{S}_n|}$, where the probability of the main player picking the vertex $P$ after knowing history $H$ is $x_{(H, P)} / x_{\sigma(H)}$, where $\sigma(H)$ is the sequence leading to the history $H$. In order for this to be a probability distribution, then \begin{align*}
    \sum_{P \in \mathfrak{C}} x_{(\emptyset, P)} &= 1 \\
    \sum_{P \text{ not queried in } H} x_{(H, P)} &= x_{\sigma(H)} \text{ for all } H \in \mathcal{H}_n \backslash \emptyset. 
\end{align*}

Finally, the payoff is represented by a matrix $M \in \mathbb{R}^{|\mathcal{S}_n|} \times \mathbb{R}^{|\Omega_n|}$, where $M_{S, \mathcal{O}}$ is the number of vertices queried in $S$ if the last vertex will lead to a sink, and otherwise it is $0$.

The optimal value is going to be $\min_x \max_y x^TMy$. Given a fixed strategy for the main player, the adversary can always pick his best strategy to be pure (in the sense that there exists some $\mathcal{O} \in \Omega_n$ such that $y_\mathcal{O} = 1$ and $y_{\mathcal{O}'} = 0$ for $\mathcal{O}' \neq \mathcal{O}$). Thus, Tessaro shows that the optimal value of the following linear program is the optimal value for a randomized algorithm for \textsc{Sink-USO} \cite[Lemma 4.7]{3rand}.

\begin{align*}
    \text{minimize } & v \\
    \text{subject to} & \sum_{S \in \mathcal{S}_n} M_{S, \mathcal{O}} \cdot x_S \leq v & \text{for all } \mathcal{O} \in \Omega_n \\
    &\sum_{P \in \mathfrak{C}} x_{(\emptyset, P)} = 1 \\
    &\sum_{P \text{ not queried in } H} x_{(H, P)} = x_{\sigma(H)} &\text{for all } H \in \mathcal{H}_n \backslash \emptyset. \\
    & x_S \geq 0 &\text{for all } S \in \mathcal{S}_n 
\end{align*}

However, we wish to solve \textsc{Sink-or-Clash} and not \textsc{Sink-USO}. Fortunately, only small modifications are needed to move from one problem to the other. It is important to redefine $\Omega_n$ to be the set of all possible outmaps, as now the adversary can choose an outmap which is not a USO. Moreover, it is essential to ensure that the histories and sequences do not have clashes (in addition to not having sinks). Finally, the payoff function is nonzero not only when the next vertex is a sink, but also when it clashes with another known vertex, as that is the other possible solution. These small modifications are enough to ensure that the linear program has the same optimal value as the best randomized algorithm for \textsc{Sink-or-Clash}, in a proof similar to the one presented in \cite[Lemma 4.7]{3rand}.

We designed a program that was able to write this linear problem, which consisted of $225$ variables, $256$ inequalities and $141$ equalities. We used Maple to solve it, giving the answer $\tilde{t}(2) = 46/17 \approx 2.706$, which is also larger than $\tilde{q}(2) = 43/20 = 2.15$ \cite[Lemma 5.1]{upperuso}.

\begin{lemma}
    $\tilde{t}(2) = \frac{46}{17}$.
\end{lemma}

Interestingly, the use of the Adapted Product Algorithm using $\tilde{t}(2) = 46/17$ as a base case yields an upper bound of $\tilde{t}(n) \in O\left(\sqrt{46/17}^n\right) = O(1.645^n)$, which is worse than the Adapted Fibonacci Seesaw.

In order to find $\tilde{t}(3)$, we use the symmetry of the $3$-hypercube, as in \cite[Subsection 4.2]{3rand}. This uses the fact that there exists an optimal strategy for the main player in which he gives the same probability to isomorphic sequences (which are sequences that can be mapped into one another using some automorphisms of the hypercube). This follows from the same proof as in \cite[Theorem 4.15]{3rand}. When using the symmetry with the $2$-dimensional hypercube, the obtained linear program only has $43$ variables, $14$ inequalities and $20$ equalities, which is a significant decrease when compared to the previous one. Unfortunately, the linear program we obtain for the $3$-dimensional hypercube consists of $87716$ variables, $352744$ inequalities and $57880$ equalities, which is too large to solve exactly using the computational resources available to us. We used Gurobi to obtain the following approximation.

\begin{lemma}
    $\tilde{t}(3) \approx 3.591333$.
\end{lemma}

When using the Adapted Product Algorithm with $\tilde{t}(3) \approx 3.591333$ as a base case, this creates an upper bound of $\tilde{t}(n) \in O\left(\sqrt[3]{3.591333}\right) = O\left(1.531^n\right)$, which is the best upper bound so far.

\begin{remark}
    The C++ code for writing this linear program taking into consideration symmetry is available at \url{https://github.com/TiagoMarques13/Non-Promise-Version-of-Unique-Sink-Orientations}. This is a variation of the implementation presented in \cite[Algorithm 1]{3rand}.
\end{remark}

\section{Resolution Proof}\label{section_resproof}

The class \texttt{TFNP} contains all total \texttt{NP} search problems, i.e., search problems where a solution is guaranteed to exist and checking such solution can be done in polynomial time. \texttt{TFNP} has been studied by defining subclasses using polynomial-time reductions to some natural complete problem \cite{10.1145/3663758}. In this paper, we only analyze the black-box model, in which the input is described as a black-box, instead of a white-box \cite{10.1145/3663758}. Formally, in the black-box model, the input can be seen as a (possibly exponentially) long bitstring, which can be accessed by queries to a random access single-bit oracle. Moreover, in this model, the algorithms and reductions are decision trees, and only the query complexity, i.e. the depth of the tree, is analyzed. This model is a natural choice for \textsc{Sink-or-Clash}, because as stated in \Cref{introduction}, we only analyze the number of queries, ignore the time between queries and we assume that the outmap is given by an oracle. The classes in this model are denoted by their name in the white-box model with a superscript \texttt{dt}, standing for ``decision tree''.

There exists a natural bijection between total search problems and unsatisfiable CNF formulas \cite[Subsection 2.1]{10.1145/3663758}. Given a total search problem, its equivalent CNF is created as a conjunction stating that every possible solution is not a solution. As a solution must exist for that total search problem, this CNF must be unsatisfiable. Vice versa, given an unsatisfiable CNF formula, it is possible to define a total search problem. Given a specific assignment of the variables (which corresponds to an instance of the problem), the problem consists in finding an unsatisfied clause, which must exist, as the CNF is unsatisfiable.

Several recent works have found connections between the query complexity of total search problems and the proof complexity of their related formulas. Hence, it is possible to assign proof systems and a measure of efficiency to some classes, so that a problem is in a certain class if and only if its CNF formula has an efficient proof in a determined proof system. A proof system for unsatisfiable CNFs can be seen as a relation $R \subseteq \{0, 1\}^* \times \{0, 1\}^*$ that is polynomial-time decidable, sound (i.e., if $(x, \mathcal{F}) \in R$, then $\mathcal{F}$ is an unsatisfiable CNF) and complete (in the sense that for any unsatisfiable CNF $\mathcal{F}$, there exists some $x \in \{0, 1\}^*$ such that $(x, \mathcal{F}) \in R$) \cite[Definition 1.5.1]{Krajíček_2019}.

The following image (adapted from \cite[Figure 2]{10.1145/3663758}) shows several classes in $\texttt{TFNP}^\texttt{dt}$ together with their corresponding proof systems.

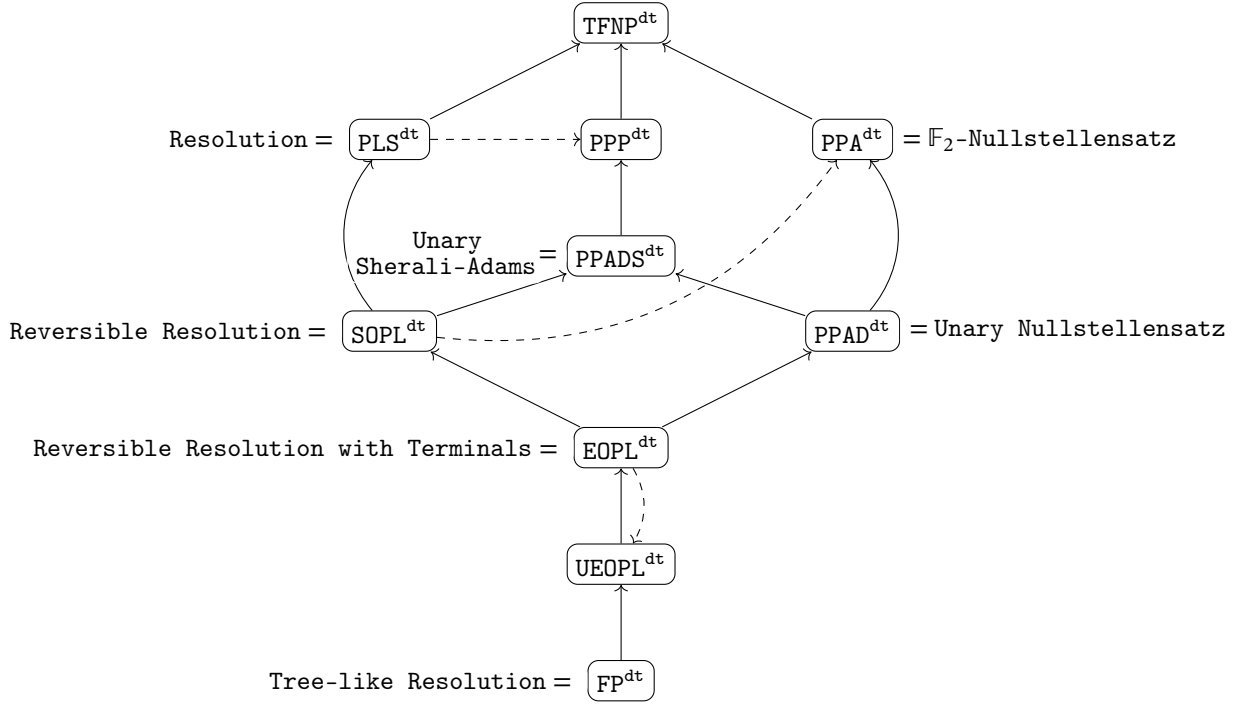
\begin{figure}[htbp] \centering
\begin{tikzpicture}[
    every node/.style={draw, rounded corners, align=center}
]
\node (TFNP) {$\texttt{TFNP}^\texttt{dt}$};
\node (PPP) [below=1cm of TFNP] {$\texttt{PPP}^\texttt{dt}$};
\node (PPA) [right=2cm of PPP] {$\texttt{PPA}^\texttt{dt}$};
\node (PPADS) [below=1cm of PPP] {$\texttt{PPADS}^\texttt{dt}$};
\node (PLS) [left=2cm of PPP] {$\texttt{PLS}^\texttt{dt}$};
\node (PPAD) [below=2cm of PPA] {$\texttt{PPAD}^\texttt{dt}$};
\node (SOPL) [below=2cm of PLS] {$\texttt{SOPL}^\texttt{dt}$};
\node (EOPL) [below=2cm of PPADS] {$\texttt{EOPL}^\texttt{dt}$};
\node (UEOPL) [below=1cm of EOPL] {$\texttt{UEOPL}^\texttt{dt}$};
\node (FP) [below=1cm of UEOPL] {$\texttt{FP}^\texttt{dt}$};
\draw[->] (PLS) -- (TFNP);
\draw[->] (PPP) -- (TFNP);
\draw[->] (PPA) -- (TFNP);
\draw[->] (SOPL) edge [bend left = 40] (PLS);
\draw[->] (PPADS) -- (PPP);
\draw[->] (SOPL) -- (PPADS);
\draw[->] (PPAD) edge [bend right = 40] (PPA);
\draw[->] (PPAD) -- (PPADS);
\draw[->] (EOPL) -- (SOPL);
\draw[->] (EOPL) -- (PPAD);
\draw[->] (UEOPL) -- (EOPL);
\draw[->] (FP) -- (UEOPL);
\draw[->, dashed] (PLS) -- (PPP);
\draw[->, dashed] (EOPL) edge [bend left] (UEOPL);
\draw[->, dashed] (SOPL) edge [bend right] (PPA);
\node[draw=white, left=0cm of PLS] {\texttt{Resolution} = };
\node[draw=white, right=0cm of PPA] { = $\mathbb{F}_2$\texttt{-Nullstellensatz}};
\node[draw=white, left=0cm of SOPL] {\texttt{Reversible Resolution} = };
\node[draw=white, right=0cm of PPAD] { = \texttt{Unary Nullstellensatz}};
\node[draw=white, left=0cm of EOPL] {\texttt{Reversible Resolution with Terminals} = };
\node[draw=white, left=0cm of FP] {\texttt{Tree-like Resolution} = };
\node[draw=white, above left=-0.37cm and 1cm of PPADS] {\texttt{Unary}};
\node[draw=white, below left=-0.37cm and 0.28cm of PPADS] {\texttt{Sherali–Adams}};
\node[draw=white, left=0cm of PPADS] {=};
\end{tikzpicture}
\caption{Representation of several classes in $\texttt{TFNP}^\texttt{dt}$. As in \cite[Figure 2]{10.1145/3663758}, a non-dashed arrow means inclusion and a dashed arrow means not inclusion. Some classes have their equivalent proof systems next to them.}
\label{img-tfnp}
\end{figure}

The problem \textsc{Sink-or-Clash} lies in $\texttt{UEOPL}^\texttt{dt}$ (Unique End of Potential Line) \cite[Theorem 24]{FEARNLEY20201}. Hence, it follows that it is also contained in $\texttt{PLS}^\texttt{dt}$, and so it has an efficient resolution proof, which we will formally define next, following the definitions in \cite[Subchapter 2.2]{10.1145/3663758} and in \cite[Chapter 5]{Krajíček_2019}.

\begin{definition}[Resolution Step]
    Given two disjunctive clauses $C$ and $D$ and a literal $x$, a \emph{resolution step} is the derivation $(C \vee x), (D \vee \neg x) \vdash (C \vee D)$.
\end{definition}

\begin{definition}[Resolution Proof]
    Given an unsatisfied CNF formula $\mathcal{F} \coloneq (C_1 \wedge \cdots \wedge C_d)$, a \emph{resolution proof} is a sequence of resolution steps which use either clauses in $\mathcal{F}$ or clauses derived in previous steps to derive new clauses, until it derives the empty class, i.e., $\bot$.
\end{definition}

\begin{definition}[Size of a Resolution Proof]
    The \emph{size} of a resolution proof is the number of resolution steps in the proof. It will be denoted by $s(\mathcal{P})$, where $\mathcal{P}$ is a proof.
\end{definition}

\begin{definition}[Width of a Resolution Proof]
    The \emph{width} of a resolution proof is defined as the largest number of literals in any clause in the proof. It will be denoted as $w(\mathcal{P})$, where $\mathcal{P}$ is a proof.
\end{definition}

\begin{definition}[Efficient Resolution Proof]
    A resolution proof $\mathcal{P}$ for a formula $\mathcal{F}$ on $N$ variables is efficient if $\log(s(\mathcal{P})) + w(\mathcal{P}) \in \text{polylog}(N)$.
\end{definition}

It is known that resolution is a proof system which is sound and complete \cite[Theorem 5.1.3]{Krajíček_2019}. This means that all derivable clauses are entailed by the initial clauses, and if the initial formula is unsatisfiable, then $\bot$ is derivable.

\begin{lemma}\cite{resproof}
    An unsatisfiable CNF formula $\mathcal{F}$ has an efficient resolution proof if and only if its search problem $\mathcal{P}$ is in $\texttt{PLS}^\texttt{dt}$.
\end{lemma}

\begin{corollary}\label{cnfresproof}
    The CNF formula of \textsc{Sink-or-Clash} has an efficient resolution proof.
\end{corollary}

However, such a resolution proof for the formula of \textsc{Sink-or-Clash} was not known yet.

\subsection{Resolution Proof of \textsc{Sink-or-Clash}}

In the specific case of \textsc{Sink-or-Clash} in an $n$-dimensional hypercube $\mathfrak{C}$, the corresponding conjunction has $n2^n$ variables, one for each bit of the outmap of each vertex. Given a vertex $u \in \mathfrak{C}$ and $i \in [n]$, we encode $s(u)_i$ as the literal $x_{u, i}$ in a way such that $x_{u, i}$ is true if and only if $s(u)_i = 1$.

Since \textsc{Sink-or-Clash} has two types of solutions (sinks or clashes), the corresponding formula has two types of clauses.
\begin{enumerate}
    \item The solution consists of a vertex $u \in \mathfrak{C}$ which is a sink. This can be identified as $\bigwedge_{i=1}^n \neg x_{u, i}$. In order for this not to be a solution, then there must exist the clause $\neg \bigwedge_{i=1}^n \neg x_{u, i} \equiv \bigvee_{i=1}^n x_{u, i}$. We will call such clauses \emph{non-sink clauses}.
    \item The solution consists of a pair of distinct vertices $u, v \in \mathfrak{C}$ which clashes. In order for this to happen, their outmap must be the same in the dimensions $I \coloneq \{i \in [n] \mid u_i \neq v_i\}$, which can be identified as $\bigvee_{o \in \{0, 1\}^I} \bigwedge_{i \in I} (x_{u, i} = o_i \wedge x_{v, i} = o_i)$. Notice that $x_{u, i} = o_i$ can be immediately translated as a literal. In order for this not to be a solution, then $\bigwedge_{o \in \{0, 1\}^I} \bigvee_{i \in I} (x_{u, i} \neq o_i \vee x_{v, i} \neq o_i)$. We will call such clauses \emph{non-clash clauses}.
\end{enumerate}

\begin{lemma}[CNF Formula of \textsc{Sink-or-Clash}]
    Given an outmap $s$ in a $n$-dimensional hypercube $\mathfrak{C}$ as an instance of \textsc{Sink-or-Clash}, the corresponding CNF formula is \[\mathcal{F}_n \coloneq \left(\bigwedge_{u \in \mathfrak{C}} \bigvee_{i \in [n]} x_{u, i} \right) \wedge \left(\bigwedge_{\substack{u, v \in \mathfrak{C} \\u \neq v}} \bigwedge_{o \in \{0, 1\}^{I_{u, v}}} \bigvee_{i \in I_{u, v}} \left(x_{u, i} \neq o_i \vee x_{v, i} \neq o_i\right)\right)\] where $I_{u, v} \coloneq \{i \in [n] \mid u_i \neq v_i\}$.
    
    As \textsc{Sink-or-Clash} is a total search problem, then $\mathcal{F}_n$ must be unsatisfiable.
\end{lemma}

\begin{example}\label{cnf1}
    For $n = 1$, the formula is \[\mathcal{F}_1 \coloneq x_{0, 1} \wedge x_{1, 1} \wedge (x_{0, 1} \vee x_{1, 1}) \wedge (\neg x_{0, 1} \vee \neg x_{1, 1}).\]

    The first two clauses claim that neither of the vertices is a sink, and the last two state that their outmap cannot be the same.
\end{example}

Let us now give a direct proof of \Cref{cnfresproof}, by constructing an efficient resolution proof of \textsc{Sink-or-Clash} recursively.

Let $\mathcal{P}_n$ denote the resolution proof for $\mathcal{F}_n$.

We begin with the induction basis, $n = 1$. The formula is $\mathcal{F}_1 = x_{0, 1} \wedge x_{1, 1} \wedge (x_{0, 1} \vee x_{1, 1}) \wedge (\neg x_{0, 1} \vee \neg x_{1, 1})$, as in \Cref{cnf1}. One initial resolution step derives $x_{0, 1}, (\neg x_{0, 1} \vee \neg x_{1, 1}) \vdash \neg x_{1, 1}$, and then finally $x_{1, 1}, \neg x_{1, 1} \vdash \bot$. This is a valid proof of $\mathcal{F}_1$, where $s(\mathcal{P}_1) = 2$ and $w(\mathcal{P}_1) = 2$.

Now, for some $n \in \mathbb{N}$ with $n \geq 2$, assume that the proof $\mathcal{P}_{n-1}$ has already been constructed. Fix some $u \in \mathfrak{C}_{1^n}^{[n-1]}$. For each $v \in \mathfrak{C}_{0^n}^{[n-1]}$, there exists a clause in $\mathcal{F}_n$ which is \[\left(\bigvee_{i \in [n-1] \colon u_i \neq v_i} (x_{u, i} \vee x_{v, i})\right) \vee \neg x_{u, n} \vee \neg x_{v, n}\] as $u_n = 1 \neq 0 = v_n$. This clause prevents $v$ and $u$ from being the sinks in the facets $\mathfrak{C}_{0^n}^{[n-1]}$ and $\mathfrak{C}_{1^n}^{[n-1]}$, respectively, but not global sinks.

Now, we derive the clause that shows that $v$ and $u$ cannot both be sinks simultaneously in $\mathfrak{C}_{0^n}^{[n-1]}$ and in $\mathfrak{C}_{1^n}^{[n-1]}$, respectively. We denote such clause as $C_{v, u} \coloneq \bigvee_{i \in [n-1]} (x_{u, i} \vee x_{v, i})$. We derive it using only two steps, as follows. \begin{align*}
    1.\;\;& \left(\bigvee_{i \in [n-1] \colon u_i \neq v_i} (x_{u, i} \vee x_{v, i})\right) \vee \neg x_{u, n} \vee \neg x_{v, n}, \bigvee_{i \in [n]} x_{u, i} \vdash \left(\bigvee_{i \in [n-1]} x_{u, i}\right) \vee \left(\bigvee_{i \in [n-1] \colon u_i \neq v_i} x_{v, i}\right) \vee \neg x_{v, n}. \\
    2.\;\;& \left(\bigvee_{i \in [n-1]} x_{u, i}\right) \vee \left(\bigvee_{i \in [n-1] \colon u_i \neq v_i} x_{v, i}\right) \vee \neg x_{v, n}, \bigvee_{i \in [n]} x_{v, i} \vdash \bigvee_{i \in [n-1]} (x_{u, i} \vee x_{v, i}).
\end{align*}

Now, focus on $\mathfrak{C}_{0^n}^{[n-1]}$. A clash in $\mathfrak{C}_{0^n}^{[n-1]}$ is also a clash in $\mathfrak{C}$, so all clauses in \[\bigwedge_{\substack{u, v \in \mathfrak{C}_{0^n}^{[n-1]} \\u \neq v}} \bigwedge_{o \in \{0, 1\}^{I_{u, v}}} \bigvee_{i \in I_{u, v}} \left(x_{u, i} \neq o_i \vee x_{v, i} \neq o_i\right)\] also appear in \[\bigwedge_{\substack{u, v \in \mathfrak{C} \\u \neq v}} \bigwedge_{o \in \{0, 1\}^{I_{u, v}}} \bigvee_{i \in I_{u, v}} \left(x_{u, i} \neq o_i \vee x_{v, i} \neq o_i\right).\]
The clauses $C_{v, u}$ for all $v \in \mathfrak{C}_{0^n}^{[n-1]}$ and a fixed $u \in \mathfrak{C}_{1^n}^{[n-1]}$ contain the non-sink clauses of $\mathcal{F}_{n-1}$, together with the extra literals $\bigvee_{i \in [n-1]} x_{u, i}$. Hence, it is possible to apply $\mathcal{P}_{n-1}$ to these clauses to solve $\mathfrak{C}_{0^n}^{[n-1]}$. By induction, at every resolution step, the used clauses either are the ones used in $\mathcal{P}_{n-1}$ or they have the extra literals $\bigvee_{i \in [n-1]} x_{u, i}$. This is true in the original clauses. Moreover, when applying a resolution step, if the used clauses contain the extra literals, they will just also appear in the derived clause. If a non-sink clause in $\mathcal{F}_{n-1}$ was not used in $\mathcal{P}_{n-1}$ for some vertex $v$, then it is impossible to derive $\bot$, as there exists a possible outmap which is a USO in which $v$ is the sink. Therefore, all non-sink clauses are used in $\mathcal{P}_{n-1}$, and so the extra literals will appear in the final clause. Hence, this recursion step derives $\bot \vee \bigvee_{i \in [n-1]} x_{u, i} = \bigvee_{i \in [n-1]} x_{u, i}$. This clause claims that $u$ is not a sink in $\mathfrak{C}_{1^n}^{[n-1]}$.

Now, perform this for each $u \in \mathfrak{C}_{1^n}^{[n-1]}$. As before, all non-clash clauses of $\mathfrak{C}_{1^n}^{[n-1]}$ also appear in $\mathcal{F}_n$, and now the non-sink clauses of $\mathfrak{C}_{1^n}^{[n-1]}$ were derived. As all clauses of $\mathcal{F}_{n-1}$ were derived, applying $\mathcal{P}_{n-1}$ will derive $\bot$, as required.

Now, we analyze the efficiency of $\mathcal{P}_n$.

\begin{lemma}
    $s(\mathcal{P}_n) \in O\left(4^{n^2}\right)$.
\end{lemma}
\begin{proof}
    In order to derive $C_{v, u}$, we use $2$ resolution steps. For a fixed $u \in \mathfrak{C}_{1^n}^{[n-1]}$, this is done for all $v \in \mathfrak{C}_{0^n}^{[n-1]}$, which takes $2 \cdot 2^{n-1}$ resolution steps. Finally, the recursion on $\mathfrak{C}_{0^n}^{[n-1]}$ takes $s(\mathcal{P}_{n-1})$ steps, to create the clause which claims that $u$ is not a sink in $\mathfrak{C}_{1^n}^{[n-1]}$. This is done for every $u \in \mathfrak{C}_{1^n}^{[n-1]}$, requiring $2^{n-1}\left(2 \cdot 2^{n-1} + s(\mathcal{P}_{n-1})\right)$ steps. In the end, $s(\mathcal{P}_{n-1})$ steps are used to evaluate $\mathfrak{C}_{1^n}^{[n-1]}$.

    Therefore, $s(\mathcal{P}_n) = s(\mathcal{P}_{n-1}) + 2^{n-1}\left(2 \cdot 2^{n-1} + s(\mathcal{P}_{n-1})\right)$, where $s(\mathcal{P}_1) = 2$. For some $n \in \mathbb{N}$, $n > 1$, if $s(\mathcal{P}_{n-1}) \leq 4^{(n-1)^2}$ (which is the case for $n = 2$), then $s(\mathcal{P}_n) \leq 4^{(n-1)^2} + 4^n + 4^{(n-1)^2+n/2} \leq 4 \cdot 4^{(n-1)^2+n/2} \leq 4^{n^2}$. Hence, it follows by induction that $s(\mathcal{P}_n) \leq 4^{n^2}$, and so $s(\mathcal{P}_n) \in O\left(4^{n^2}\right)$, as claimed.   
\end{proof}

\begin{lemma}
    $w(\mathcal{P}_n) \in O\left(n^2\right)$.
\end{lemma}
\begin{proof}
    The initial clauses in $\mathcal{F}_n$ contain at most $2n$ literals each. Then, $C_{v, u}$ contains $2n-2$ literals. When performing $\mathcal{P}_{n-1}$ in $\mathfrak{C}_{0^n}^{[n-1]}$, the maximum number of literals in any clause is going to be $w(\mathcal{P}_{n-1}) + n-1$, due to the extra $n-1$ literals. Finally, when applying $\mathcal{P}_{n-1}$ to $\mathfrak{C}_{1^n}^{[n-1]}$, there are no extra literals, so the width there is $w(\mathcal{P}_{n-1})$. The initial case is $w(\mathcal{P}_1) = 2$. Therefore, $w(\mathcal{P}_n) = \max(2n, w(\mathcal{P}_{n-1}) + n - 1)$ for $n > 1$, which solves to \[w(\mathcal{P}_n) = \begin{cases}
    2n & \text{if } n \leq 3 \\
    \frac{n^2-n+6}{2} & \text{if } n \geq 4 \\ 
    \end{cases}\]
    In general, $w(\mathcal{P}_n) \in O\left(n^2\right)$.
\end{proof}

Hence, $\log(s(\mathcal{P}_n)) + w(\mathcal{P}_n) \in O(n^2) \subset \text{polylog}(n2^n)$, where $n2^n$ is the total number of variables. Therefore, this is an efficient resolution proof that solves \textsc{Sink-or-Clash}.

\begin{remark}
    In this resolution proof, only one permutation of the dimensions was considered in the recursion (from dimension $1$ to $n$). Hence, this resolution proof also efficiently proves the formula corresponding to a $\texttt{UEOPL}^\texttt{dt}$-complete version of \textsc{Cube-OPDC} (which is presented in \cite[Theorem 20]{FEARNLEY20201}). This shows that $\texttt{UEOPL}^\texttt{dt} \subseteq  \texttt{PLS}^\texttt{dt}$, which was already known as in \Cref{img-tfnp}.
\end{remark}

\section{Completability of Unique Sink Orientations}\label{cert-chapter}

According to \Cref{notharder}, finding a sink or a clash in an outmap is not easier than finding the sink of a USO. Moreover, let $\mathcal{A}$ be the optimal algorithm for \textsc{Sink-USO}. When this algorithm is used for \textsc{Sink-or-Clash}, if it fails to find the sink, it must be because the orientation induced by the vertices already queried is not compatible with any USO. Hence, in this chapter, we will analyze the minimal subsets of outmaps that are not compatible with any USO (defined formally in \Cref{cert-def}), specially the ones in which no clash exist.

Initially, some definitions must be made.

\begin{definition}[Partial Outmap]
    A function $s \colon \mathfrak{C} \to \{0, 1\}^n \cup \{*\}$ is called a \emph{partial outmap} of the $n$-dimensional hypercube.
\end{definition}

Notice that this definition of partial outmap is only taking into consideration the outmap of each vertex. It is possible that this does not define a proper orientation in $\mathfrak{C}$, as one edge might be oriented in different directions depending on the endpoints.

\begin{definition}[Agreeable Partial Outmaps]
    A partial outmap $t$ of the $n$-dimensional hypercubes $\mathfrak{C}$ \emph{agrees} with another partial outmap $s$ of $\mathfrak{C}$ if for all $u \in \mathfrak{C}$, $s(u) \neq \;* \implies s(u) = t(u)$.
\end{definition}

\begin{definition}[Completable Partial Outmap]
    A function $s \colon \mathfrak{C} \to \{0, 1\}^n \cup \{*\}$ is called a \emph{completable partial outmap} if there exists a unique sink orientation $t \colon \mathfrak{C} \to \{0, 1\}^n$ which agrees with $s$.
\end{definition}

\begin{definition}[Non-completable Partial Outmap]
    A function $s \colon \mathfrak{C} \to \{0, 1\}^n \cup \{*\}$ is called a \emph{non-completable partial outmap} if it is not a completable partial outmap.
\end{definition}

\begin{definition}[Clash]
    Given a partial outmap $s \colon \mathfrak{C} \to \{0, 1\}^n \cup \{*\}$, a \emph{clash} consists of two distinct vertices $u, v \in \mathfrak{C}$ such that $s(u) \neq \; *$, $s(v) \neq \; *$ and $(s(u) \oplus s(v)) \wedge (u \oplus v) = 0^n$. Notice that this definition of clash agrees with the preexisting one in the case $s$ is an outmap.
\end{definition}

A completable partial outmap can be seen as a partial outmap function of a USO, where some of the values have been hidden. Hence, if $s$ is a completable partial outmap then it follows that no pair of vertices clashes. However, this condition is not enough to ensure that a partial outmap is completable, as seen in \Cref{4-cert}.

\begin{figure}[htbp]\centering
\begin{tikzpicture}
        \draw[middlearrow={<}] (0,0,0) -- (2,0,0);
        \draw[middlearrow={<}] (2,0,0) -- (2,2,0);
        \draw[middlearrow={<}] (2,2,0) -- (0,2,0);
        \draw[middlearrow={>}] (0,2,0) -- (0,0,0);
        \draw[middlearrow={<}] (0,0,2) -- (2,0,2);
        \draw[middlearrow={>}] (2,0,2) -- (2,2,2);
        \draw[middlearrow={<}] (2,2,2) -- (0,2,2);
        \draw[middlearrow={<}] (0,2,2) -- (0,0,2);
        \draw[middlearrow={<}] (0,0,0) -- (0,0,2);
        \draw[middlearrow={>}] (2,0,0) -- (2,0,2);
        \draw[middlearrow={<}] (0,2,0) -- (0,2,2);
        \draw[middlearrow={>}] (2,2,0) -- (2,2,2);
        \fill [blue] (2, 0, 0) circle (2pt);
        \fill [blue] (0, 2, 0) circle (2pt);
        \fill [blue] (0, 0, 2) circle (2pt);
        \fill [blue] (2, 2, 2) circle (2pt);
\end{tikzpicture}
    \caption{A non-completable partial outmap in which no pair clashes. The blue vertices are the vertices that have been assigned an outmap. All the edges are determined but there are two sinks.}
    \label{4-cert}
\end{figure}
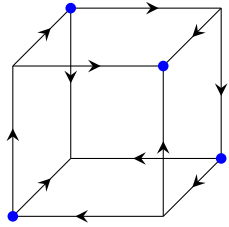

It is possible to define the decision problem \textsc{Uso-Completability}, which has access to an oracle that encodes a partial outmap $s$ and returns YES if and only if $s$ is completable. \textsc{Uso-Completability} is not easier than \textsc{Uso-Recognition} \cite[Section 3]{gartnerrec} (the problem of deciding whether a certain outmap is a USO), as \textsc{Uso-Completability} solves \textsc{Uso-Recognition} when the input is an outmap. Hence, as \textsc{Uso-Recognition} is \texttt{coNP}-hard, it follows that \textsc{Uso-Completability} is also \texttt{coNP}-hard (when the oracle is represented as a Boolean circuit).

Gärtner and Thomas use clashes as an efficiently verifiable certificate that an outmap is not a USO. Moreover, due to \Cref{fund}, any outmap that is not a USO has a clash, so this shows that \textsc{Uso-Recognition} is in \texttt{coNP} \cite[Section 3]{gartnerrec}. However, this proof does not show that \textsc{Uso-Completability} is in \texttt{coNP}, as some non-completable partial outmaps have no clashes. Furthermore, it is not known whether \textsc{Uso-Completability} is in \texttt{coNP} or not. So, as clashes are not enough to show non-completability, we study the nature of minimal non-completable partial outmaps (defined formally as \emph{certificates} in \Cref{cert-def}).

Note that in the previous sections when we studied \textsc{Sink-or-Clash}, we could assume access to an oracle to \textsc{Uso-Completability} without cost, since we only cared about the query complexity.

\subsection{Certificates for Non-Completability}

\begin{definition}[Known Vertices]
    Given a partial outmap $s \colon \mathfrak{C} \to \{0, 1\}^n \cup \{*\}$, define the set of known vertices of $s$ as $P_s \coloneq \{u \in \mathfrak{C} \mid s(u) \neq *\}$.
\end{definition}

\begin{definition}
    Given a partial outmap $s \colon \mathfrak{C} \to \{0, 1\}^n$ and $P \subseteq P_s$, define the partial outmap $s_P \colon \mathfrak{C} \to \{0, 1\}^n$ as \[s_P(u) = \begin{cases}
        s(u) & \text{if } u \in P \\
        * & \text{if } u \notin P\\
    \end{cases}\]
\end{definition}

Notice that in the previous definition, $s$ agrees with $s_P$. This can be seen as forgetting some of the already known outmaps.

\begin{definition}[$k$-certificate]\label{cert-def}
    A partial outmap $s \colon \mathfrak{C} \to \{0, 1\}^n \cup \{*\}$ is considered a \emph{$k$-certificate} if $|P_s| = k$, $s$ is non-completable and for all $P \subset P_s$ the partial outmap $s_P$ is completable. Sometimes, the notation \emph{certificate} is used for a $k$-certificate for some $k \in \mathbb{N}$, without specifying $k$.
\end{definition}

\begin{observation}
    A partial outmap $s$ is non-completable if and only if there exists some $P \subseteq P_s$ such that $s_P$ is a certificate.
\end{observation}

If a partial outmap $s$ only assigns the outmap of a single vertex (i.e., $|P_s| = 1$), then it is always completable (by using any USO and \Cref{flipall}, per example). Moreover, Schurr showed that if $|P_s| = 2$ and the pair of vertices does not clash, then the partial outmap is completable to a USO \cite[Lemma 4.21]{schurrthesis}. So, the $2$-certificates consist of two vertices which clash. Furthermore, Schurr showed that given the outmaps of $3$ vertices which do not pairwise clash, then that is always completable to a USO, which we phrase as follows.

\begin{lemma}\label{non3}\cite[Lemma 4.22]{schurrthesis}
    There are no $3$-certificates.
\end{lemma}

Next, we show that although there are no $3$-certificates, there are $n$-certificates for any $n \geq 4$.

\begin{lemma}
    There exist $n$-certificates for all $n \in \mathbb{N}$, $n \geq 4$.
\end{lemma}
\begin{proof}
    In this proof, assume that the indices are seen modulo $n$, i.e., $u_{n+1} \coloneq u_1$ per example.
    
    Consider the partial outmap $s \colon \mathfrak{C} \to \{0, 1\}^n \cup \{*\}$ defined as 
    \begin{itemize}
        \item $s(1^n) = 0^n$.
        \item $s(e_i) = e_i \vee e_{i+1}$ for every $i \in [n]$.
        \item $s(u) = \; *$ for the remaining vertices $u \in \mathfrak{C}$.
    \end{itemize}
    
    We claim that this is a $(n+1)$-certificate, for every $n \geq 3$.

    Suppose $s$ is completable. As $0^n$ and $e_i$ do not clash, it follows that $(s(0^n) \oplus s(e_i)) \wedge (e_i \oplus 0^n) \neq 0 \Rightarrow s(0^n)_i = 0$. As this is true for all $1 \leq i \leq n$, then $s(0^n) = 0^n$. Therefore, $0^n$ and $1^n$ clash, which is a contradiction. Hence, $s$ is non-completable.

    Consider any proper subset $P \subset P_s$. It is enough to show that $s_P$ is completable when $|P| = n = |P_s| - 1$. We check two different cases.

    \begin{enumerate}[label=\textbf{Case \arabic*:}, leftmargin=*]
       \item
       Suppose that $P = P_s \backslash \{1^n\}$. Consider the uniform orientation $t$ given by $t(u) = u$. Then, all edges are flippable. Moreover, $\{(e_i, e_i \vee e_{i+1}) \mid 1 \leq i \leq n\}$ is a matching of the edges, which when flipped create a USO $t'$ \cite[Lemma 4.30]{schurrthesis}. As $t'(e_i) = t(e_i) \oplus e_{i+1} = e_i \oplus e_{i+1} = e_i \vee e_{i+1}$, then $t'$ is a USO which agrees with $s_P$. Thus, $s_P$ is completable.
       
       \item
        Otherwise, suppose without loss of generality that $P = P_s \backslash \{e_n\}$. It follows by induction that $s_P$ is constructible.

        If $n = 3$, then $|P| = 3$. Moreover, $(e_i \oplus e_j) \wedge (s(e_i) \oplus s(e_j)) = e_i \oplus e_j$ if $|j-i| \neq 1$, $(e_i \oplus e_{i+1}) \wedge (s(e_i) \oplus s(e_{i+1})) = e_i$ and $(e_i \oplus 1^n) \wedge (s(e_i) \oplus s(1^n)) = e_{i+1}$, so there are no clashes. By \Cref{non3}, $s_P$ is completable.
        
        Now, assume that $n > 3$. Let $s'$ be a partial outmap of $\mathfrak{C}$ such that 
        \begin{itemize}
            \item $s'(e_i) = e_i \vee e_{i+1}$ for $1 \leq i < n$.
            \item $s'(1^n) = 0^n$.
            \item $s'(e_i \vee e_n) = e_i \vee e_{i+1} \vee e_n$ for $1 \leq i < n-1$.
            \item $s'(u) = \;*$ for all other $u \in \mathfrak{C}$.
        \end{itemize}

        Notice that for all $u \in P$, $s(u) = s'(u)$, so $s'$ agrees with $s_P$. Therefore, if $s'$ is completable, then $s_P$ is also completable.
        
        Consider the facets $\mathfrak{C}_{0^n}^{[n-1]}$ and $\mathfrak{C}_{1^n}^{[n-1]}$.
        
        Notice that $P_{s'} \cap \mathfrak{C}_{0^n}^{[n-1]} = \{e_i \mid 1 \leq i < n\}$. Similarly to Case $1$, let $t$ be the uniform outmap $t(u) = u$ for $u \in \mathfrak{C}$ and then let $t'$ be the outmap obtained after flipping the matching $\{(e_i, e_i \vee e_{i+1}) \mid 1 \leq i < n-1\}$ of $t$. Then, $t'(e_i)_{[n-1]} = (e_i \vee e_{i+1})_{[n-1]} = s'(e_i)$ for $i < n-1$ and $t'(e_{n-1})_{[n-1]} = s'(e_{n-1})_{[n-1]} = e_{n-1}$. Thus, this facet is completable.
        
        As $P_{s'} \cap \mathfrak{C}_{1^n}^{[n-1]} = \{e_i \vee e_n \mid 1 \leq i < n-1\} \cup \{1^n\}$, it follows that this facet is also completable by induction within $\mathfrak{C}_{1^n}^{[n-1]}$, as it corresponds to the $n-1$ case.

        After completing both facets, orient all $n$-edges from $\mathfrak{C}_{0^n}^{[n-1]}$ to $\mathfrak{C}_{1^n}^{[n-1]}$, with the exception of the edges $(e_i, e_i \oplus e_n)$ for $1 \leq i < n-1$, which are oriented in the other direction. Let this final outmap be $r$, which agrees with $s'$.

        Notice that $r(e_i)_{[n-1]} = s'(e_i)_{[n-1]} = s'(e_i \vee e_n)_{[n-1]} = r(e_i \vee e_n)_{[n-1]}$ for $1 \leq i < n-1$, so the edges $(e_i, e_i \oplus e_n)$ are flippable according to \Cref{flip}. Let $r'$ be the outmap obtained after flipping all such edges. As those edges form a matching, it follows that $r$ is a USO if and only if $r'$ is a USO. Moreover, $r'$ is combed in the dimension $n$, and as it is both a USO when restricted to $\mathfrak{C}_{0^n}^{[n-1]}$ and $\mathfrak{C}_{1^n}^{[n-1]}$, then $r'$ is a USO according to \Cref{product}.

        Hence, $s'$ is completable, and it follows that $s_P$ is completable.
    \end{enumerate}
\end{proof}

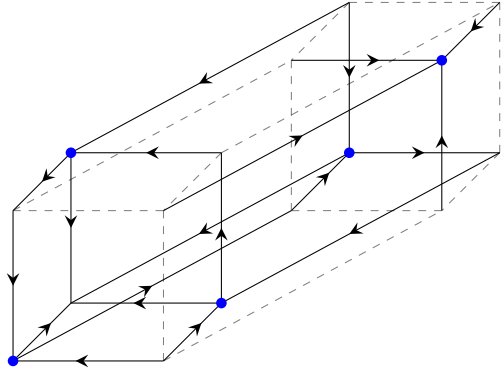
\begin{figure}[htbp]\centering
\begin{tikzpicture}
        \coordinate (offset) at (\offsetx, \offsety, \offsetz);
            
        \draw[middlearrow={<}] (0,0,0) -- (2,0,0);
        \draw[middlearrow={>}] (2,0,0) -- (2,2,0);
        \draw[middlearrow={>}] (2,2,0) -- (0,2,0);
        \draw[middlearrow={>}] (0,2,0) -- (0,0,0);
        \draw[middlearrow={<}] (0,0,2) -- (2,0,2);
        \draw[dashed, opacity=0.5] (2,0,2) -- (2,2,2);
        \draw[dashed, opacity=0.5] (2,2,2) -- (0,2,2);
        \draw[middlearrow={>}] (0,2,2) -- (0,0,2);
        \draw[middlearrow={<}] (0,0,0) -- (0,0,2);
        \draw[middlearrow={<}] (2,0,0) -- (2,0,2);
        \draw[middlearrow={>}] (0,2,0) -- (0,2,2);
        \draw[dashed, opacity=0.5] (2,2,0) -- (2,2,2);

        \draw[middlearrow={>}] ($(0,0,0)+(offset)$) -- ($(2,0,0) + (offset)$);
        \draw[dashed, opacity=0.5] ($(2,0,0)+(offset)$) -- ($(2,2,0)+(offset)$);
        \draw[dashed, opacity=0.5] ($(2,2,0)+(offset)$) -- ($(0,2,0)+(offset)$);
        \draw[middlearrow={>}] ($(0,2,0)+(offset)$) -- ($(0,0,0)+(offset)$);
        \draw[dashed, opacity=0.5] ($(0,0,2)+(offset)$) -- ($(2,0,2)+(offset)$);
        \draw[middlearrow={>}] ($(2,0,2)+(offset)$) -- ($(2,2,2)+(offset)$);
        \draw[middlearrow={<}] ($(2,2,2)+(offset)$) -- ($(0,2,2)+(offset)$);
        \draw[dashed, opacity=0.5] ($(0,2,2)+(offset)$) -- ($(0,0,2)+(offset)$);
        \draw[middlearrow={<}] ($(0,0,0)+(offset)$) -- ($(0,0,2)+(offset)$);
        \draw[dashed, opacity=0.5] ($(2,0,0)+(offset)$) -- ($(2,0,2)+(offset)$);
        \draw[dashed, opacity=0.5] ($(0,2,0)+(offset)$) -- ($(0,2,2)+(offset)$);
        \draw[middlearrow={>}] ($(2,2,0)+(offset)$) -- ($(2,2,2)+(offset)$);

        \draw[middlearrow={<}] (0,0,0) -- ($(0,0,0)+(offset)$);
        \draw[middlearrow={<}] (2,0,0) -- ($(2,0,0)+(offset)$);
        \draw[dashed, opacity=0.5] (2,2,0) -- ($(2,2,0)+(offset)$);
        \draw[middlearrow={<}] (0,2,0) -- ($(0,2,0)+(offset)$);
        \draw[middlearrow={>}] (0,0,2) -- ($(0,0,2)+(offset)$);
        \draw[dashed, opacity=0.5] (2,0,2) -- ($(2,0,2)+(offset)$);
        \draw[middlearrow={>}] (2,2,2) -- ($(2,2,2)+(offset)$);
        \draw[dashed, opacity=0.5] (0,2,2) -- ($(0,2,2)+(offset)$);
            
        \fill [blue] (2,0,0) circle (2pt);
        \fill [blue] (0,2,0) circle (2pt);
        \fill [blue] (0,0,2) circle (2pt);
        \fill [blue] (offset) circle (2pt);
        \fill [blue] ($(2,2,2)+(offset)$) circle (2pt);
\end{tikzpicture}
    \caption{\label{5-cert}A $5$-certificate.}
\end{figure}

This fully determines for which $n \in \mathbb{N}$ there exists an $n$-certificate.

\subsection{4-certificates}

In this final subsection, we fully characterize and categorize the the $4$-certificates.

\begin{lemma}\label{extend}
    Let $s \colon \mathfrak{C} \to \{0, 1\}^n$ be a partial outmap. Let $u, v \in \mathfrak{C}$ be any two distinct vertices that do not clash and such that $(s(u) \oplus s(v)) \wedge (u \oplus v) \neq e_i$ for some fixed $i \in [n]$. Then, $u \oplus e_i \neq v$ and if $s(u \oplus e_i) = s(u) \oplus e_i$, it follows that $u \oplus e_i$ and $v$ do not clash.
\end{lemma}
\begin{proof}
    As $u$ and $v$ do not clash, it follows that there exists some $j \in [n]$ such that $u_j \neq v_j \wedge s(u)_j \neq s(v)_j$. Assume $j \neq i$ as $(s(u) \oplus s(v)) \wedge (u \oplus v)) \neq e_i$. Therefore, $u_j \neq v_j$, which implies that $u \oplus e_i \neq v$.  Moreover, $(u \oplus e_i)_j = u_j \neq v_j$ and $s(u \oplus e_i)_j = s(u)_j \neq s(v)_j$, so $u \oplus e_i$ and $v$ also do not clash.    
\end{proof}

\Cref{fund} shows that given any pair of vertices which do not clash, there must exist some dimension in which their outmap is different. Then, \Cref{extend} states that we can collapse any dimension as long as it is not the only one in which they differ, and the resulting vertices still do not clash.

\begin{definition}[Projected Partial Outmap]
    Let $s$ be a partial outmap of the $n$-dimensional $\mathfrak{C}$. Let $I \subseteq [n]$. Define the partial outmap $s^I \colon \{0, 1\}^{|I|} \to \{0, 1\}^{|I|} \cup \{*\}$ such that $s^I(u_I) = s(u)_I$ for all $u \in P_s$ and $s^I(v) = \;*$ for all other $v \in \{0, 1\}^{|I|}$. This is called a \emph{projected partial outmap} to the dimensions $I$. Notice that this is only well-defined when $s(u)_I = s(v)_I$ whenever $u_I = v_I$ and $u, v \in P_s$.
\end{definition}

Notice that this projected partial outmap can be interpreted as removing some dimensions of the hypercube and merging those lost dimensions together. Furthermore, the order in which the dimensions are removed is irrelevant.

\begin{definition}[Spanned Dimensions]
    Let $s$ be a partial outmap of $\mathfrak{C}$. A dimension $i \in [n]$ is called \emph{spanned} if there are $u, v \in P_S$ such that $u_i \neq v_i$. Otherwise, it is called \emph{non-spanned}.
\end{definition}

\begin{lemma}[Spanned Dimensions]\label{span}
    Let $s$ be a partial outmap of $\mathfrak{C}$ and $i \in [n]$ a non-spanned dimension. Then, $s$ is completable if and only if $s^{[n] \backslash \{i\}}$ is completable.
\end{lemma}
\begin{proof}
    Without loss of generality, assume that $i = n$, by reordering the dimensions.
    
    Suppose that $u_{[n-1]} = v_{[n-1]}$ for some $u, v \in \mathfrak{C}$. As $n$ is not spanned, then $u_n = v_n$ and $u = v$. Hence, $s^{[n-1]}$ is well-defined.

    Pick any $u \in P_s$. As $n$ is non-spanned, then $P_s \subseteq \mathfrak{C}_u^{[n-1]}$.

    Assume that $s$ is completable. Then, there exists a USO $t \colon \mathfrak{C} \to \{0, 1\}^n$ which agrees with $s$. Moreover, the restriction of the orientation $t$ to the face $\mathfrak{C}_u^{[n-1]}$ is also a USO, which agrees with $s^{[n-1]}$, as $P_s \subseteq \mathfrak{C}_u^{[n-1]}$. Thus, $s^{[n-1]}$ is completable.

    Assume that $s^{[n-1]}$ is completable. Then, let $t$ be a USO in the facet $\mathfrak{C}_u^{[n-1]}$ which agrees with $s^{[n-1]}$. Construct the outmap $r \colon \mathfrak{C} \to \{0, 1\}^n$ by creating two copies of $t$ in the opposite facets $\mathfrak{C}_u^{[n-1]}$ and $\mathfrak{C}_{u \oplus e_n}^{[n-1]}$, and orienting the $n$-edges in the direction so that it agrees with $s$ in the $n$-edges. By the product construction in \ref{product}, $r$ is a USO no matter how the $n$-edges are oriented. Furthermore, $r$ agrees with $s$ in the facet $\mathfrak{C}_u^{[n-1]}$, and it also agrees in the $n$-edges, by definition, so $s$ is completable.
\end{proof}

Multiple uses of \Cref{span} show that given a partial outmap $s$ and the subset of the spanned dimensions $I \subseteq [n]$, then $s$ is completable if and only if $s^I$ is completable.

\begin{definition}[Projectable Dimensions]
    Let $s$ be a partial outmap of $\mathfrak{C}$. A dimension $i \in [n]$ is called \emph{projectable} if $s(u)_i = s(v)_i$ for all $u, v \in P_s$. Otherwise, it is called \emph{non-projectable}.
\end{definition}

\begin{lemma}[Projectable Dimensions]\label{proj}
    Let $s$ be a partial outmap of $\mathfrak{C}$ with no clashes and a projectable dimension $i \in [n]$. Then, $s$ is completable if and only if $s^{[n]\backslash\{i\}}$ is completable.
\end{lemma}
\begin{proof}
    Without loss of generality, assume that $i = n$, by reordering the dimensions. Due to \Cref{flipall}, flipping the orientation of all $n$-edges yields an equivalent problem, so assume without loss of generality that $s(u)_n = 0$ for all $u \in P_s$.

    Suppose that $u_{[n-1]} = v_{[n-1]}$ for some distinct $u, v \in P_s$. As $u$ and $v$ do not clash, it follows that $u_n \neq v_n \wedge s(u)_n \neq s(v)_n$, which is a contradiction. Thus, $s^{[n-1]}$ is well-defined.

    Suppose that $s^{[n-1]}$ is completable. Then, there exists an outmap $t$ of an $(n-1)$-dimensional USO which agrees with $s^{[n-1]}$. Create a new outmap $t'$ by using two copies of $t$ in $\mathfrak{C}_{0^n}^{[n-1]}$ and $\mathfrak{C}_{1^n}^{[n-1]}$, and orienting the $n$-edges such that $t'(u)_n = s(u)_n$ and $t'(u \oplus e_n)_n = 1 - s(u)_n$ for each $u \in P_s$ (and the others can have any orientation). Notice that this is always possible, because if $u \in P_s$, then $u \oplus e_n \notin P_S$. By \Cref{product}, the obtained orientation is a USO. Finally, for each $u \in P_s$, it follows that $t'(u)_{[n-1]} = t(u_{[n-1]})_{[n-1]} = s^{[n-1]}(u_{[n-1]}) = s(u)_{[n-1]}$ and $t'(u)_n = s(u)_n$, so $t$ agrees with $s$. Hence, $s$ is completable.

    Suppose that $s$ is completable. Then, let $t \colon \mathfrak{C} \to \{0, 1\}^n$ be the outmap of a USO which agrees with $s$. Perform a partial swap as in \Cref{partswap} in the dimension $n$ in order to create another USO and restrict it afterwards to $\mathfrak{C}_{0^n}^{[n-1]}$. Call the resulting $(n-1)$-dimensional USO $t'$. For any $u \in P_s$, then $s^{[n-1]}(u_{[n-1]}) = s(u)_{[n-1]} = t(u)_{[n-1]}$. If $u_n = 0$, then it will not change after the partial swap and $u \in \mathfrak{C}_{0^n}^{[n-1]}$, so $t(u)_{[n-1]} = t'(u_{[n-1]})$. If $u_n = 1$, then it will change during the partial swap and $u \oplus e_n \in \mathfrak{C}_{0^n}^{[n-1]}$, so $t(u)_{[n-1]} = t'((u \oplus e_n)_{[n-1]}) = t'(u_{[n-1]})$. Either way, $s^{[n-1]}(u_{[n-1]}) = t'(u_{[n-1]})$, so $s^{[n-1]}$ is completable.    
\end{proof}

Similarly, \Cref{proj} shows that it is also possible to ignore any projectable dimension. Multiple uses of \Cref{proj} state that it is possible to ignore all non-projectable dimensions.

In particular, when analyzing a $4$-certificate $s$, we can assume that there no projectable dimensions and that $P_s$ spans all the dimensions of the hypercube, as otherwise we could efficiently remove those dimensions using \Cref{span} and \Cref{proj}, and it would reduce to a case with fewer dimensions.

\begin{theorem}
    Let $s \colon \mathfrak{C} \to \{0, 1\}^n \cup \{*\}$ be a $4$-certificate with no projectable dimensions and such that $P_s$ spans all dimensions. Then, $n = 3$ and $s$ is isomorphic to either \Cref{4-cert1} or \Cref{4-cert2}, possibly by flipping all the edges in some dimensions.
\end{theorem}

\begin{figure}[htbp]
    \centering
    \begin{subfigure}[t]{0.45\textwidth}
        \centering
        \begin{tikzpicture}
            \draw[middlearrow={<}] (0,0,0) -- (2,0,0);
            \draw[middlearrow={<}] (2,0,0) -- (2,2,0);
            \draw[middlearrow={<}] (2,2,0) -- (0,2,0);
            \draw[middlearrow={>}] (0,2,0) -- (0,0,0);
            \draw[middlearrow={<}] (0,0,2) -- (2,0,2);
            \draw[middlearrow={>}] (2,0,2) -- (2,2,2);
            \draw[middlearrow={<}] (2,2,2) -- (0,2,2);
            \draw[middlearrow={<}] (0,2,2) -- (0,0,2);
            \draw[middlearrow={<}] (0,0,0) -- (0,0,2);
            \draw[middlearrow={>}] (2,0,0) -- (2,0,2);
            \draw[middlearrow={<}] (0,2,0) -- (0,2,2);
            \draw[middlearrow={>}] (2,2,0) -- (2,2,2);
            \fill [blue] (2, 0, 0) circle (2pt);
            \fill [blue] (0, 2, 0) circle (2pt);
            \fill [blue] (0, 0, 2) circle (2pt);
            \fill [blue] (2, 2, 2) circle (2pt);
        \end{tikzpicture}
        \caption{4-certificate}
        \label{4-cert1}
    \end{subfigure}
    \hfill
    \begin{subfigure}[t]{0.45\textwidth}
        \centering
        \begin{tikzpicture}
            \draw[middlearrow={<}] (0,0,0) -- (2,0,0);
            \draw[middlearrow={<}] (2,0,0) -- (2,2,0);
            \draw[red, dashed] (2,2,0) -- (0,2,0);
            \draw[middlearrow={<}] (0,2,0) -- (0,0,0);
            \draw[red, dashed] (0,0,2) -- (2,0,2);
            \draw[middlearrow={>}] (2,0,2) -- (2,2,2);
            \draw[middlearrow={<}] (2,2,2) -- (0,2,2);
            \draw[middlearrow={>}] (0,2,2) -- (0,0,2);
            \draw[middlearrow={<}] (0,0,0) -- (0,0,2);
            \draw[middlearrow={>}] (2,0,0) -- (2,0,2);
            \draw[middlearrow={<}] (0,2,0) -- (0,2,2);
            \draw[middlearrow={>}] (2,2,0) -- (2,2,2);
            \fill [blue] (0, 0, 0) circle (2pt);
            \fill [blue] (2, 0, 0) circle (2pt);
            \fill [blue] (0, 2, 2) circle (2pt);
            \fill [blue] (2, 2, 2) circle (2pt);
        \end{tikzpicture}
        \caption{4-certificate}
        \label{4-cert2}
    \end{subfigure}
    \caption{All $4$-certificates up to automorphisms and flipping all edges in a subset of the dimensions.}
\end{figure}
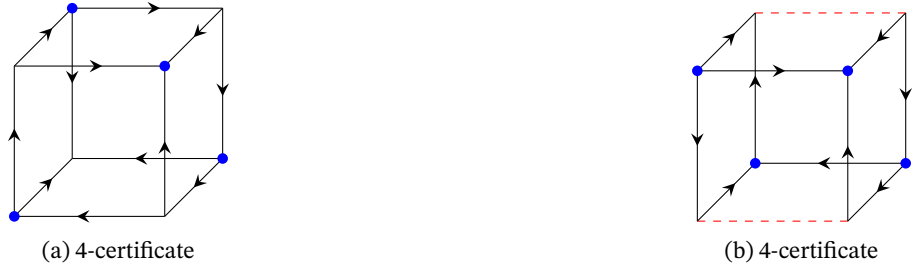

\begin{proof}
    Let the $4$ vertices of the certificate be $u, v, w, z \in P_s$. As $s$ is a $4$-certificate, then no pair of those vertices clashes. Notice that $n = 1$ is impossible and if $n = 2$, then all vertices are determined, so $s$ would be already a USO, which is a contradiction. Assume now for the rest of the proof that $n \geq 3$.

    Denote a dimension $i \in [n]$ by \emph{special} if (up to renaming of the vertices) we have \begin{itemize}
        \item $u_i = v_i = 0$ and $w_i = z_i = 1$,
        \item $s(u)_i = s(w)_i = 0$ and $s(v)_i = s(z)_i = 1$, and
        \item $(s(u) \oplus s(z)) \wedge (u \oplus z) = e_i$ and $(s(v) \oplus s(w)) \wedge (v \oplus w) = e_i$.
    \end{itemize}
    
    In each special dimension $i$, there are $2$ pairs of vertices of $S_p$ which only avoid clashing due to the dimension $i$. Moreover, there are only $6$ pairs of vertices, so there are at most $3$ special dimensions.

    Assume that $n > 3$. Then, there exists some dimension which is not special. Assume without loss of generality that the dimension $n$ is not special. There are some cases to analyze.

    \begin{enumerate}[label=\textbf{Case \arabic*:}, leftmargin=*, labelsep=12pt]
        \item $\left|P_s \cap \mathfrak{C}_{0^n}^{[n-1]}\right| \neq \left|P_s \cap \mathfrak{C}_{1^n}^{[n-1]}\right|$.
        
        As all dimensions are spanned, we can assume without loss of generality that $u, v, w \in \mathfrak{C}_{0^n}^{[n-1]}$ and $z \in \mathfrak{C}_{1^n}^{[n-1]}$. According to \Cref{flipall}, flipping all $n$-edges maintains the USO property, so further assume without loss of generality that $s(z)_n = 0$. There are still some possible configurations, depending on the number of $n$-edges of $u, v, w$ that are oriented in each direction.

        \begin{enumerate}[label=\textbf{Case \arabic{enumi}.\arabic*:}, leftmargin=0pt, labelsep=5pt]
            \item $s(u)_n = s(v)_n = s(w)_n = 0$.

            In this case, dimension $n$ would be projectable, which is a contradiction.

            \item\label{casecombed} $s(u)_n = s(v)_n = s(w)_n = 1$.

            Create a USO $t$ in $\mathfrak{C}_{0^n}^{[n-1]}$ which agrees with $u, v, w$. This is always possible due to \Cref{non3}. Create a USO $t'$ in $\mathfrak{C}_{1^n}^{[n-1]}$ which agrees with $z$. Combine them into a USO $r$ by orienting all $n$-edges towards $\mathfrak{C}_{1^n}^{[n-1]}$. Notice that $r$ is a USO due to the product construction in \Cref{product}. Moreover, $r$ agrees with $s$ in the first $n-1$ dimensions, due to the construction of $t$ and $t'$, and also in the $n$-edges. Hence, $s$ is completable, which is a contradiction.

            \begin{figure}[htbp]
                \centering
                \begin{tikzpicture}

                \coordinate (u) at (-2.1,0);
                \coordinate (v) at (-0.7,0);
                \coordinate (w) at (0.7,0);
                \coordinate (z) at (2.1,2);
    
                \draw[] (0,0) ellipse (3cm and 0.8cm);
                \draw[] (0,2) ellipse (3cm and 0.8cm);
                \node[] at (3.6,0) {$\mathfrak{C}_{0^n}^{[n-1]}$};
                \node[] at (3.6,2) {$\mathfrak{C}_{1^n}^{[n-1]}$};

                \fill [black] (u) circle (2pt);
                \fill [black] (v) circle (2pt);
                \fill [black] (w) circle (2pt);
                \fill [black] (z) circle (2pt);
                \node[below] at (u) {$u$};
                \node[below] at (v) {$v$};
                \node[below] at (w) {$w$};
                \node[above] at (z) {$z$};

                \draw[middlearrow={>}] (u) -- ($(u) + (0,2)$);
                \draw[middlearrow={>}] (v) -- ($(v) + (0,2)$);
                \draw[middlearrow={>}] (w) -- ($(w) + (0,2)$);
                \draw[middlearrow={<}] (z) -- ($(z) - (0,2)$);
                
                \end{tikzpicture}
                \caption{Case 1.2}
                \label{case1.2}
            \end{figure}
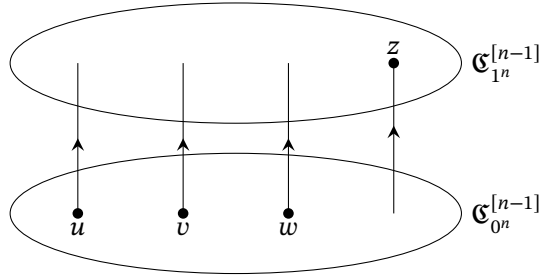

            \item Let $A \coloneq \{x \in \{u, v, w\} \mid s(x)_n = 0\}$ and $B \coloneq \{u, v, w\} \backslash A$. Assume $|A|, |B| > 0$.

            Initially, create a USO $t$ in $\mathfrak{C}_{0^n}^{[n-1]}$ which agrees with $u, v, w$. As in the last case, this is always possible.

            For each $x \in A$, extend $s$ by defining $s(x \oplus e_n) = s(x) \oplus e_n$. By \Cref{extend}, then $x \oplus e_n$ and $z$ do not clash, as $x$ and $z$ do not clash and $s(x)_n = s(z)_n$. Moreover, for distinct $x, y \in A$, as $x$ and $y$ do not clash, there exists some $j \in [n]$ such that $x_j \neq y_j \wedge s(x)_j \neq s(y)_j$. As $x_n = y_n$, then $j \neq n$, and so $(x \oplus e_n)_j \neq (y \oplus e_n)_j$ and $s(x \oplus e_n)_j = s(x)_j \neq s(y)_j = s(y \oplus e_n)_j$, so $x \oplus e_n$ and $y \oplus e_n$ also do not clash.

            As $0 \leq |A| \leq 2$, then at most $3$ vertices in $\mathfrak{C}_{1^n}^{[n-1]}$ are known. As they do not pairwise clash, there is some USO $t'$ in $\mathfrak{C}_{1^n}^{[n-1]}$ which agrees with all of them, by \Cref{non3}.

            Finally, construct $r$ by combining both $t$ and $t'$, and orienting the $n$-edges towards $\mathfrak{C}_{1^n}^{[n-1]}$. Notice that $r$ is a USO by the product construction in \Cref{product}. Furthermore, $r$ agrees with $u, v, w, z$ everywhere but in the $n$-edges for each $x \in A$. Nevertheless, for each $x \in A$, it follows that $s(x)_{[n-1]} = s(x \oplus e_n)_{[n-1]}$, so those $n$-edges are flippable and form a matching. After flipping them, the resulting USO agrees with $s$, so $s$ is completable, which is a contradiction.   

            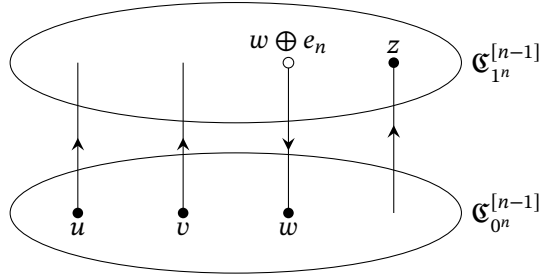
\begin{figure}[htbp]
                \centering
                \begin{tikzpicture}

                \coordinate (u) at (-2.1,0);
                \coordinate (v) at (-0.7,0);
                \coordinate (w) at (0.7,0);
                \coordinate (z) at (2.1,2);
    
                \draw[] (0,0) ellipse (3cm and 0.8cm);
                \draw[] (0,2) ellipse (3cm and 0.8cm);
                \node[] at (3.6,0) {$\mathfrak{C}_{0^n}^{[n-1]}$};
                \node[] at (3.6,2) {$\mathfrak{C}_{1^n}^{[n-1]}$};

                \fill [black] (u) circle (2pt);
                \fill [black] (v) circle (2pt);
                \fill [black] (w) circle (2pt);
                \fill [black] (z) circle (2pt);
                \node[below] at (u) {$u$};
                \node[below] at (v) {$v$};
                \node[below] at (w) {$w$};
                \node[above] at (z) {$z$};

                \draw[middlearrow={>}] (u) -- ($(u) + (0,2)$);
                \draw[middlearrow={>}] (v) -- ($(v) + (0,2)$);
                \draw[middlearrow={<}] (w) -- ($(w) + (0,2)$);
                \draw[middlearrow={<}] (z) -- ($(z) - (0,2)$);

                \node[above] at ($(w)+(0,2)$) {$w \oplus e_n$};
                \draw [black, fill=white] ($(w)+(0,2)$) circle (2pt);
                
                \end{tikzpicture}
                \caption{Case 1.3}
                \label{case1.3}
            \end{figure}
        \end{enumerate}

        \item $\left|P_s \cap \mathfrak{C}_{0^n}^{[n-1]}\right| = \left|P_s \cap \mathfrak{C}_{1^n}^{[n-1]}\right| = 2$.

        Assume without loss of generality that $u, v \in \mathfrak{C}_{0^n}^{[n-1]}$ and $w, z \in \mathfrak{C}_{1^n}^{[n-1]}$. There are some possible configurations, after assuming that the dimension $n$ is not special.

        \begin{enumerate}[label=\textbf{Case \arabic{enumi}.\arabic*:}, leftmargin=0pt, labelsep=5pt]
            \item $s(u)_n = s(v)_n$ or $s(w)_n = s(z)_n$.

            Assume without loss of generality that $s(u)_n = s(v)_n = 0$, by also possibly flipping the direction of all $n$-edges, according to \Cref{flipall}. As $n$ is not projectable, assume further that $s(w)_n = 1$.

            If $s(z)_n = 1$, all given $n$-edges are oriented towards $\mathfrak{C}_{0^n}^{[n-1]}$. Complete the USO in $\mathfrak{C}_{0^n}^{[n-1]}$ agreeing with $u$ and $v$, then complete the USO in $\mathfrak{C}_{1^n}^{[n-1]}$ agreeing with $w$ and $z$, and finally orient all $n$-edges towards $\mathfrak{C}_{0^n}^{[n-1]}$. This is possible as no pair of vertices clash, and the resulting orientation is a combed USO according to \Cref{product}. As it agrees with $s$, it follows that $s$ is completable, which is a contradiction.

            \begin{figure}[htbp]
                \centering
                \begin{tikzpicture}

                \coordinate (u) at (-2.1,0);
                \coordinate (v) at (-0.7,0);
                \coordinate (w) at (0.7,2);
                \coordinate (z) at (2.1,2);
    
                \draw[] (0,0) ellipse (3cm and 0.8cm);
                \draw[] (0,2) ellipse (3cm and 0.8cm);
                \node[] at (3.6,0) {$\mathfrak{C}_{0^n}^{[n-1]}$};
                \node[] at (3.6,2) {$\mathfrak{C}_{1^n}^{[n-1]}$};

                \fill [black] (u) circle (2pt);
                \fill [black] (v) circle (2pt);
                \fill [black] (w) circle (2pt);
                \fill [black] (z) circle (2pt);
                \node[below] at (u) {$u$};
                \node[below] at (v) {$v$};
                \node[above] at (w) {$w$};
                \node[above] at (z) {$z$};

                \draw[middlearrow={<}] (u) -- ($(u) + (0,2)$);
                \draw[middlearrow={<}] (v) -- ($(v) + (0,2)$);
                \draw[middlearrow={>}] (w) -- ($(w) - (0,2)$);
                \draw[middlearrow={>}] (z) -- ($(z) - (0,2)$);
                
                \end{tikzpicture}
                \caption{First Part of Case 2.1}
                \label{case2.1A}
            \end{figure}
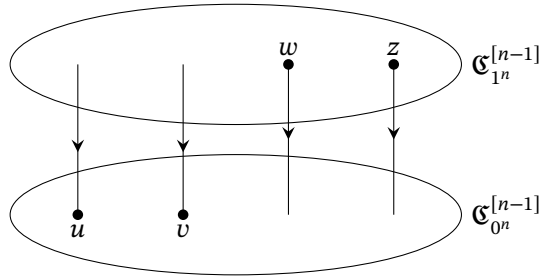

            If $s(z)_n = 0$, extend $s$ by defining $s(z \oplus e_n) = s(z) \oplus e_n$. According to \Cref{extend}, then $z \oplus e_n$ and $u$, and $z \oplus e_n$ and $v$ do not clash, as $s(z)_n = s(u)_n = s(v)_n$. Create a USO $t$ in $\mathfrak{C}_{0^n}^{[n-1]}$ which agrees with $u, v$ and $z \oplus e_n$, which is always possible according to \Cref{non3} as they do not pairwise clash. Create a USO $t'$ in $\mathfrak{C}_{1^n}^{[n-1]}$ which agrees with $w$ and $z$, which is possible as they do not clash. Finally, combine them together in an orientation $r$ by orienting all $n$-edges towards $\mathfrak{C}_{0^n}^{[n-1]}$. It follows that $r$ is a combed USO, according to \Cref{product}. Moreover, $r$ agrees with $s$ in all edges except for $(z, z \oplus e_n)$. As $s(z)_{[n-1]} = s(z \oplus e_n)_{[n-1]}$, then that edge is flippable according to \Cref{flip}, so after flipping it, the result is a USO which agrees with $s$. Therefore, $s$ is completable, which is a contradiction.

            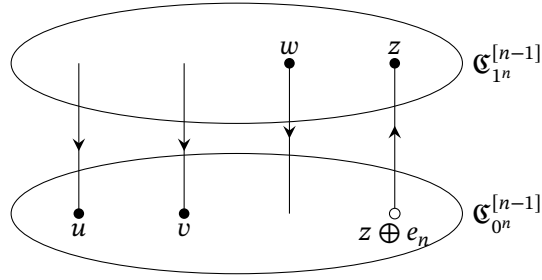
\begin{figure}[htbp]
                \centering
                \begin{tikzpicture}

                \coordinate (u) at (-2.1,0);
                \coordinate (v) at (-0.7,0);
                \coordinate (w) at (0.7,2);
                \coordinate (z) at (2.1,2);
    
                \draw[] (0,0) ellipse (3cm and 0.8cm);
                \draw[] (0,2) ellipse (3cm and 0.8cm);
                \node[] at (3.6,0) {$\mathfrak{C}_{0^n}^{[n-1]}$};
                \node[] at (3.6,2) {$\mathfrak{C}_{1^n}^{[n-1]}$};

                \fill [black] (u) circle (2pt);
                \fill [black] (v) circle (2pt);
                \fill [black] (w) circle (2pt);
                \fill [black] (z) circle (2pt);
                \node[below] at (u) {$u$};
                \node[below] at (v) {$v$};
                \node[above] at (w) {$w$};
                \node[above] at (z) {$z$};

                \draw[middlearrow={<}] (u) -- ($(u) + (0,2)$);
                \draw[middlearrow={<}] (v) -- ($(v) + (0,2)$);
                \draw[middlearrow={>}] (w) -- ($(w) - (0,2)$);
                \draw[middlearrow={<}] (z) -- ($(z) - (0,2)$);

                \node[below] at ($(z)-(0,2)$) {$z \oplus e_n$};
                \draw[black, fill=white] ($(z)-(0,2)$) circle (2pt);
                
                \end{tikzpicture}
                \caption{Second Part of Case 2.1}
                \label{case2.1B}
            \end{figure}

            \item Assume without loss of generality that $s(u)_n = s(w)_n = 0$ and $s(v)_n = s(z)_n = 1$ and $(s(u) \oplus s(z)) \wedge (u \oplus z) \neq e_n$. Hence, there exists some $j \in [n-1]$ such that $u_j \neq z_j \wedge s(u)_j \neq s(z)_j$.

            Extend $s$ by defining $s(u \oplus e_n) = s(u) \oplus e_n$. According to \Cref{extend}, $u$ and $w$ do not clash (as $s(u)_n = s(w)_n$) and $u \oplus e_n$ and $z$ do not clash.

            Similarly, extend $s$ by defining $s(z \oplus e_n) = s(z) \oplus e_n$. By the same reasons as before, $z \oplus e_n$ and $u$, and $z \oplus e_n$ and $v$ do not clash.

            Construct a USO $t$ in $\mathfrak{C}_{0^n}^{[n-1]}$ which agrees with $u, v$ and $z \oplus e_n$ and a USO $t'$ in $\mathfrak{C}_{1^n}^{[n-1]}$ which agrees with $w, z$ and $u \oplus e_n$. This is possible according to \Cref{non3}, as the vertices do not pairwise clash. Finally, combine them together in the orientation $r$, by orienting all the $n$-edges towards $\mathfrak{C}_{1^n}^{[n-1]}$. It follows that $r$ is a combed USO, according to \Cref{product}.

            Moreover, $r$ agrees with $s$ in all edges except for $(u, u \oplus e_n)$ and $(z, z \oplus e_n)$. As $s(u)_{[n-1]} = s(u \oplus e_n)_{[n-1]}$ and $s(z)_{[n-1]} = s(z \oplus e_n)_{[n-1]}$, then those edges are flippable, so after flipping them, the resulting USO agrees with $s$. Thus, $s$ would be completable, which is a contradiction.

            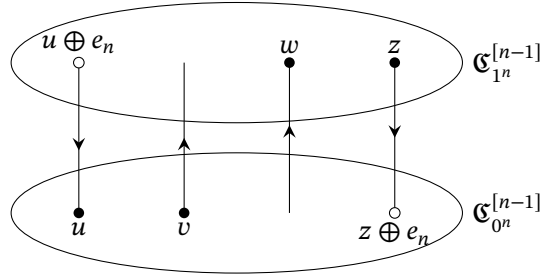
\begin{figure}[htbp]
                \centering
                \begin{tikzpicture}

                \coordinate (u) at (-2.1,0);
                \coordinate (v) at (-0.7,0);
                \coordinate (w) at (0.7,2);
                \coordinate (z) at (2.1,2);
    
                \draw[] (0,0) ellipse (3cm and 0.8cm);
                \draw[] (0,2) ellipse (3cm and 0.8cm);
                \node[] at (3.6,0) {$\mathfrak{C}_{0^n}^{[n-1]}$};
                \node[] at (3.6,2) {$\mathfrak{C}_{1^n}^{[n-1]}$};

                \fill [black] (u) circle (2pt);
                \fill [black] (v) circle (2pt);
                \fill [black] (w) circle (2pt);
                \fill [black] (z) circle (2pt);
                \node[below] at (u) {$u$};
                \node[below] at (v) {$v$};
                \node[above] at (w) {$w$};
                \node[above] at (z) {$z$};

                \draw[middlearrow={<}] (u) -- ($(u) + (0,2)$);
                \draw[middlearrow={>}] (v) -- ($(v) + (0,2)$);
                \draw[middlearrow={<}] (w) -- ($(w) - (0,2)$);
                \draw[middlearrow={>}] (z) -- ($(z) - (0,2)$);

                \node[above] at ($(u)+(0,2)$) {$u \oplus e_n$};
                \draw[black, fill=white] ($(u)+(0,2)$) circle (2pt);
                \node[below] at ($(z)-(0,2)$) {$z \oplus e_n$};
                \draw[black, fill=white] ($(z)-(0,2)$) circle (2pt);
                
                \end{tikzpicture}
                \caption{Case 2.2}
                \label{case2.2}
            \end{figure}
        \end{enumerate}
    \end{enumerate}

    This shows that $n = 3$. Furthermore, all $3$ dimensions must be special, otherwise $s$ would be completable, as stated above. Now, there are two possible cases.

    \begin{enumerate}[label=\textbf{Case \arabic*:}, leftmargin=*, labelsep=12pt]

    \item No pair of vertices in $P_s$ is adjacent.

    There exists only one configuration of vertices (shown in \Cref{pusoimage}). Moreover, such configuration determine all edges in $\mathfrak{C}$. Each facet is a square with two antipodal vertices in $S_p$, so it must be a USO, otherwise the antipodal vertices would clash. Thus, in order for $s$ not to be completable, it would follow that the obtained orientation is not a USO, but all its facets are. These are called Pseudo USOs (PUSO)~\cite[Definition 2]{puso}.

    In a PUSO, only the antipodal vertices clash \cite[Theorem 5]{puso}, so no pair of vertices in $S_p$ clashes. Hence, in any PUSO, selecting this configuration of $4$ vertices yields a $4$-certificate.

    There are only $2$ different $3$-dimensional PUSOs \cite[Figure 3]{puso}, which are obtained from one another after flipping all the edges in one dimension. Hence, in this case, all $4$-certificates are obtained from \Cref{pusoimage} after flipping all the edges in some subset of the dimensions.

    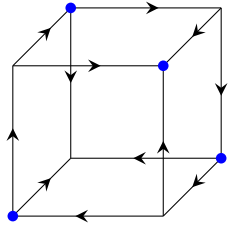
\begin{figure}[htbp]\centering
    \begin{tikzpicture}
            \draw[middlearrow={<}] (0,0,0) -- (2,0,0);
            \draw[middlearrow={<}] (2,0,0) -- (2,2,0);
            \draw[middlearrow={<}] (2,2,0) -- (0,2,0);
            \draw[middlearrow={>}] (0,2,0) -- (0,0,0);
            \draw[middlearrow={<}] (0,0,2) -- (2,0,2);
            \draw[middlearrow={>}] (2,0,2) -- (2,2,2);
            \draw[middlearrow={<}] (2,2,2) -- (0,2,2);
            \draw[middlearrow={<}] (0,2,2) -- (0,0,2);
            \draw[middlearrow={<}] (0,0,0) -- (0,0,2);
            \draw[middlearrow={>}] (2,0,0) -- (2,0,2);
            \draw[middlearrow={<}] (0,2,0) -- (0,2,2);
            \draw[middlearrow={>}] (2,2,0) -- (2,2,2);
            \fill [blue] (2, 0, 0) circle (2pt);
            \fill [blue] (0, 2, 0) circle (2pt);
            \fill [blue] (0, 0, 2) circle (2pt);
            \fill [blue] (2, 2, 2) circle (2pt);
    \end{tikzpicture}
        \caption{Pseudo USO}
        \label{pusoimage}
    \end{figure}

    \item A pair of vertices is adjacent. Assume $u \oplus e_1 = v$. As each facet must have $2$ vertices in $P_s$, the other two vertices must also satisfy $w \oplus e_1 = z$ and they must be $2$ pairs of antipodal vertices, i.e., $u \oplus z = v \oplus w = 1^n$. 

    As the two known $1$-edges must have different orientations, assume that $s(u)_1 = s(z)_1 = 1$ and $s(v)_1 = s(w)_1 = 0$.

    As $u$ and $z$ do not clash, assume that $s(u)_2 = 1 \neq s(z)_2 = 0$. As the dimension $2$ is special, then $s(v)_2 = 0$ and $s(w)_2 = 1$.

    As flipping all $3$-edges is possible, assume $s(u)_3 = 1$. Hence, $s(v)_3 = 0$. In order for $u$ and $w$, and $v$ and $z$ not to clash, then $s(w)_3 = 0$ and $s(z)_3 = 1$.

    This fully defines the outmap of all $4$ vertices. Moreover, the edges $(u \oplus e_3, v \oplus e_3)$ and $(u \oplus e_2, v \oplus e_2)$, which are red in \Cref{4-cert2image}, have no possible orientations.

    Moreover, in this configuration, no pair of vertices clash, $s$ is not completable and there are no $3$-certificates. Therefore, this is a $4$-certificate.

    \begin{figure}[htbp]
        \centering
        \begin{subfigure}[t]{0.45\textwidth}
            \centering
            \begin{tikzpicture}
                \draw[middlearrow={>}] (0,0,0) -- (2,0,0);
                \draw[middlearrow={<}] (2,0,0) -- (2,2,0);
                \draw[red, dashed] (2,2,0) -- (0,2,0);
                \draw[middlearrow={<}] (0,2,0) -- (0,0,0);
                \draw[red, dashed] (0,0,2) -- (2,0,2);
                \draw[middlearrow={>}] (2,0,2) -- (2,2,2);
                \draw[middlearrow={>}] (2,2,2) -- (0,2,2);
                \draw[middlearrow={>}] (0,2,2) -- (0,0,2);
                \draw[middlearrow={>}] (0,0,0) -- (0,0,2);
                \draw[middlearrow={<}] (2,0,0) -- (2,0,2);
                \draw[middlearrow={>}] (0,2,0) -- (0,2,2);
                \draw[middlearrow={<}] (2,2,0) -- (2,2,2);
                \fill [blue] (0, 0, 0) circle (2pt);
                \fill [blue] (2, 0, 0) circle (2pt);
                \fill [blue] (0, 2, 2) circle (2pt);
                \fill [blue] (2, 2, 2) circle (2pt);
                \node[above right] at (0,0,0) {$u$};
                \node[above right] at (2,0,0) {$v$};
                \node[above left] at (0,2,2) {$w$};
                \node[above left] at (2,2,2) {$z$};
            \end{tikzpicture}
            \caption{4-certificate}
            \label{4-cert2image}
        \end{subfigure}
        \hfill
        \begin{subfigure}[t]{0.45\textwidth}
            \centering
            \begin{tikzpicture}
                \draw[-] (0,0,0) -- (2,0,0);
                \draw[-] (0,0,0) -- (0,2,0);
                \draw[-] (0,0,0) -- (0,0,2);
                \node at (2.3,0,0) {$1$};
                \node at (0,2.3,0) {$2$};
                \node at (0,0,2.5) {$3$};
            \end{tikzpicture}
            \caption{Labelling of the Dimensions}
            \label{labeldimension}
        \end{subfigure}
        \caption{$4$-certificate}
    \end{figure}
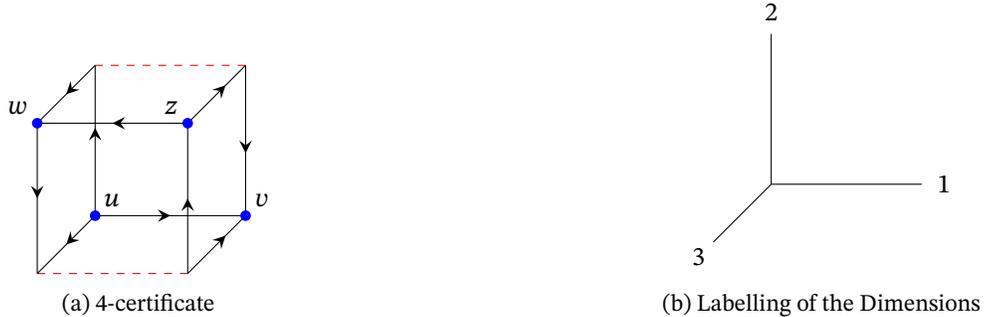
    \end{enumerate}
\end{proof}

\section{Future Work}

Studying the problem \textsc{Sink-or-Clash} seems to be an important step in investigating USOs, as several problems that reduce to USO are promise problems, as seen in \Cref{introduction}. Moreover, there are many aspects of \textsc{Sink-or-Clash} that should be studied further.

Although it seems that \textsc{Sink-or-Clash} is harder than \textsc{Sink-USO} in the deterministic setting, this is not the case up until $4$-dimensional hypercubes. However, we have not been able to find upper bounds for \textsc{Sink-USO} or lower bounds for \textsc{Sink-or-Clash} that confirm that \textsc{Sink-or-Clash} is indeed strictly harder.

In \Cref{section_resproof}, a resolution proof of \textsc{Sink-or-Clash} was constructed. Although efficient proofs in weaker proof systems are known to exist, for example using \emph{Reversible Resolution}, the proof system that characterizes $\texttt{SOPL}^\texttt{dt}$ (as $\textsc{Sink-or-Clash} \in \texttt{UEOPL}^\texttt{dt} \subseteq \texttt{SOPL}^\texttt{dt}$ \cite{FEARNLEY20201, 10.1145/3663758}), we do not know concrete proofs. Hence, the ideas behind the resolution proof of \textsc{Sink-or-Clash} could potentially be used to find them or to design an intuitive proof system that would solve the class of problems in which \textsc{Sink-or-Clash} is complete. Since the given resolution proof also works for \textsc{Cube-OPDC}, which is $\texttt{UEOPL}^\texttt{dt}$-complete \cite[Theorem 20]{FEARNLEY20201}, it seems more likely to be used to construct a proof system characterizing $\texttt{UEOPL}^\texttt{dt}$.

We do not have a categorization of the $n$-certificates for $n > 4$ (which may not exist). However, even with a good characterization, it is not clear how to find an upper bound on the number of queries needed to find a clash given a certain certificate. This would be an upper bound between the number of queries needed for \textsc{Sink-or-Clash} and \textsc{Sink-USO}, and it would help to relate the two problems. Moreover, the $4$-certificates are efficiently verifiable, so maybe one can define a non-promise version of USO which also accepts them as a solution, and this problem could potentially be easier than \textsc{Sink-or-Clash}.

Moreover, an interesting open question would be whether in a fixed dimension $d$, there would only be $n$-certificates up to a small $n=f(d)$, or if all $n$-certificates without non-spanned or projectable dimensions only exist in hypercubes up to $d = g(n)$ dimensions (as this is the case for $2$-certificates and $4$-certificates).

\newpage
\bibliographystyle{plainurl}
\bibliography{Bibliography}

\end{document}